\documentclass[11pt]{article}
\usepackage[a4paper,margin=1in]{geometry}
\usepackage[T1]{fontenc}
\usepackage[dvips]{graphicx}
\graphicspath{{images/}}

\usepackage{verbatim}
\usepackage{amsmath,amsthm}
\usepackage{xspace}
\usepackage{pifont}
\usepackage{graphicx}
\usepackage{amssymb}
\usepackage{epic, eepic}
\usepackage{dsfont}
\usepackage{amssymb}
\usepackage{makeidx}
\usepackage{mathrsfs}
\usepackage{exscale}
\usepackage{color} 
\usepackage{overpic} 
\usepackage{bm}
\usepackage{bbm}
\usepackage{booktabs} 
\usepackage{xcolor, colortbl}
\usepackage{subcaption}
\usepackage[numbers]{natbib}
\usepackage{enumerate}

\RequirePackage[colorlinks,citecolor=blue,urlcolor=blue]{hyperref}
\usepackage{natbib}
\definecolor{Gray}{gray}{0.9}

\usepackage{amsmath,afterpage}
\usepackage{epsf}
\usepackage{graphics,color}
\RequirePackage[colorlinks,citecolor=blue,urlcolor=blue]{hyperref}
\usepackage[export]{adjustbox}

\def\0{\mathbf{0}}

\def\rr{\rightarrow}
\def\dr{\downarrow}

\def \< {\langle}
\def \> {\rangle}

\def\beqa{\begin{eqnarray}}
\def\eeqa{\end{eqnarray}}
\def\beqas{\begin{eqnarray*}}
\def\eeqas{\end{eqnarray*}}

\newtheorem{theorem}{Theorem}[section]
\newtheorem{lemma}[theorem]{Lemma}
\newtheorem{prop}[theorem]{Proposition}

\newtheorem{corollary}[theorem]{Corollary}

\newtheorem{remark}[theorem]{Remark}
\newtheorem{example}[theorem]{Example}

\newtheorem{definition}[theorem]{Definition}

\numberwithin{equation}{section}
\newcommand{\hatd}[1]{{}}

\newcommand{\bd}{\begin{displaymath}}
\newcommand{\ed}{\end{displaymath}}
\newcommand{\be}{\begin{equation}}
\newcommand{\ee}{\end{equation}}
\newcommand{\bq}{\begin{eqnarray}}
\newcommand{\eq}{\end{eqnarray}}
\newcommand{\bn}{\begin{eqnarray*}}
\newcommand{\en}{\end{eqnarray*}}

\def\wt{\widetilde}

\def\P{\mathbb{P}}
\def\C{\mathbb{C}}
\def\E{{\mathbb{E}}}
\newcommand{\R}{{\mathbb R}}
\newcommand{\N}{{\mathbb N}}

\usepackage[normalem]{ulem}

\usepackage{authblk}

\usepackage{mathtools}
\mathtoolsset{showonlyrefs}
\definecolor{blue0}{RGB}{0,77,153} 
\definecolor{red0}{RGB}{179,0,77} 
\definecolor{green0}{RGB}{134,219,76} 
\definecolor{gray0}{RGB}{84,97,110}

\title{Optimal Portfolio Choice with Cross-Impact Propagators }
\author[1]{Eduardo Abi Jaber\thanks{eduardo.abi-jaber@polytechnique.edu, EAJ is grateful for the financial support from the Chaires FiME-FDD, Financial Risks, Deep Finance \& Statistics and Machine Learning \& Systematic Methods in Finance at Ecole \mbox{Polytechnique}.}}
\author[2]{Eyal Neuman\thanks{e.neumann@imperial.ac.uk}}
\author[2]{Sturmius Tuschmann\thanks{s.tuschmann23@imperial.ac.uk, ST is supported by the EPSRC Centre for Doctoral Training in Mathematics of Random \mbox{Systems}: Analysis, Modelling and Simulation (EP/S023925/1).}}
\affil[1]{Ecole Polytechnique, CMAP}
\affil[2]{Department of Mathematics, Imperial College London}

\begin{document}
\vspace{-0.5cm}
\maketitle

\begin{abstract}
\noindent We consider a class of optimal portfolio choice problems in continuous time where the agent's transactions create both transient cross-impact driven by a matrix-valued Volterra propagator, as well as temporary price impact. We formulate this problem as the maximization of a revenue-risk functional, where the agent also exploits available information on a progressively measurable price predicting signal. We solve the maximization problem explicitly in terms of operator resolvents, by reducing the corresponding first-order condition to a coupled system of stochastic Fredholm equations of the second kind and deriving its solution. We then give sufficient conditions on the matrix-valued propagator so that the model does not permit price manipulation. We also provide an implementation of the solutions to the optimal portfolio choice problem and to the associated optimal execution problem. Our solutions  yield financial insights on the influence of cross-impact on the optimal strategies and its interplay with  alpha decays.
\end{abstract} 

\begin{description}
\item[Mathematics Subject Classification (2020):] 93E20, 60H30, 91G80
\item[JEL Classification:] C02, C61, G11
\item[Keywords:] cross-impact, optimal portfolio choice, price impact, propagator models, predictive signals, Volterra stochastic control
\end{description}





\section{Introduction}

Optimal portfolio choice has been a core problem in quantitative finance. In this class of problems investors dynamically select their portfolio while taking into account expected returns, risk factors, transaction and price impact costs, and trading signals. The seminal papers of \citet{Garleanu_13,garleanu2016dynamic} provided a tractable framework for formulating and explicitly solving these problems under the assumptions of quadratic trading costs and exponential decay of the price impact, among others. Their work was further generalized in various directions (see e.g. \citep{ekren2019portfolio,horst2019multi,tsoukalas2019dynamic}). Despite its phenomenal impact, the main drawback of Gârleanu and Pedersen's framework is that it relies on the assumption that the price impact and the cross-impact decay at an exponential rate. 
 
 Price impact refers to the empirical fact that the execution of a large order
affects a risky asset's price in an adverse and persistent manner, leading to less favorable prices. Propagator models express price moves in terms of the influence of past trades, and therefore capture the decay of the price impact after each trade \citep{bouchaud_bonart_donier_gould_2018,bouchaud2009markets,curato2017optimal, gatheral2010no,webster2023handbook}. Cross-impact models provide an additional explanation for the price dynamics of a risky asset in terms of the influence of past trades of other assets in the market. The price distortion $D_t$ due to both effects is quantified by the $N$-dimensional process,   
\be \label{d_t} 
D_t  =\int_0^t G(t,s) dX_s, \quad t \geq 0, 
\ee
where $X_t=(X^1_t,...,X_t^N)$ describes the amount of shares in each of the $N$ assets in the portfolio at time $t$. Here $G(t,s)$ is a matrix of Volterra kernel functions, so that self-impact is captured by its diagonal entries and cross-impact is captured by its off-diagonal entries. The components of $G$ are often referred to as propagators (see e.g. \citep{benzaquen2017dissecting,bouchaud_bonart_donier_gould_2018,coz2023cross, mastromatteo2017trading,schneider2019cross,tomas2022build}). Throughout this paper we will adopt the convention introduced in Chapter~14.5.3 of \cite{bouchaud_bonart_donier_gould_2018} and refer to price impact (or equivalently price distortion) as the aggregated effects of self-impact and cross-impact. 

Gârleanu and Pedersen along with follow-up papers assumed that the price distortion $D_t$ decays exponentially with time. 
On the other hand, \citet{bouchaud_bonart_donier_gould_2018} (see Chapter 14.5.3) report on two main empirical observations regarding price impact: (i) Diagonal and off-diagonal elements of the propagator matrix $G$ decay as a power-law of the lag, i.e. $G_{ij}(t,s) \approx (t-s)^{-\beta_{ij}}$, with $0<\beta_{ij} < 1$. (ii) Most of the cross-correlations between price moves ($\approx 60-90\%$, depending on the timescale) are mediated by trades themselves, i.e. through a cross-impact mechanism, rather than through the cross-correlation of noise terms, which are not directly related to trading. See also \cite{benzaquen2017dissecting} for further details. 
\citet{coz2023cross} perform an empirical analysis using data of 500 US assets from 2017 to 2022. They conclude that price formation occurs endogenously within highly liquid assets first, and then trades in these assets influence the prices of their less liquid correlated products, with an impact speed constrained by their minimum trading frequency.

While empirical evidence on cross-impact and its power-law behavior is mixed \cite{benzaquen2017dissecting,bouchaud_bonart_donier_gould_2018,capponi2020multi,cont2023cross}, results on the corresponding optimal portfolio choice problem have been scarce for two main reasons: (i) Power-law type or any more general propagator matrices which do not have an exponential decay turn the optimal portfolio choice problem into non-Markovian and often time inconsistent, hence standard tools of stochastic control such as dynamic programming or FBSDEs do not apply in this case. (ii) Cross-impact may cause various problems from the financial point of view, such as round trips or
transaction-triggered price manipulations (see \cite{alfonsi2016multivariate,huberman2004price,schneider2019cross} for some pathological examples). Some recent progress has been made in \citep{abijaber2022optimal,abijaber2023equilibrium} for the single-asset version of this problem in the context of optimal execution  with signals. However the solution for the single-asset problem clearly does not take into account the aforementioned cross-impact effects.  

The main theoretical contribution in this non-Markovian portfolio choice setting was derived by \citet{alfonsi2016multivariate}, who proposed a discrete-time model for portfolio liquidation in presence of linear transient cross-impact using convolution type decay kernels $G(t)$. They characterized conditions on $G$ which guarantee that the model admits well-behaved optimal execution strategies by ensuring that price manipulation in the sense of \citet{huberman2004price} is excluded.  This means that a portfolio starting and terminating with zero inventory $(X_0=X_T=0)$ cannot create any profit, or equivalently negative trading costs, due to the induced price distortion $D_t$ in \eqref{d_t}. This condition translates into the following inequality, 
\be \label{round-t}
\E  \left[  \int_0^T D_t dX_t \right]   \geq 0,
\ee
where the time grid is discrete in \cite{alfonsi2016multivariate}, so that the integral is written as a sum. \citet{alfonsi2016multivariate} concluded that $G$ must be a matrix-valued nonnegative definite function in order for the model to satisfy condition \eqref{round-t} and thus preclude price manipulation. A sufficient condition for \eqref{round-t} in which $G$ is nonincreasing, nonnegative, convex, symmetric and commuting was given in \cite{alfonsi2016multivariate}, as well as a characterization of matrix-valued nonnegative definite functions.  Alfonsi et al. also derived a first-order condition for the optimal strategy of the execution problem and obtained an explicit solution for the case where the propagator matrix is given by $G(t)=\exp(-tC)$ for a symmetric nonnegative definite matrix $C\in\R^{N\times N}$. Note that in the framework of \cite{alfonsi2016multivariate}, not only the model is discrete in time, but also the alpha signals and the risk terms, which are central features of portfolio choice problems, are not incorporated. These simplifying assumptions turn the resulting stochastic control problem into a completely deterministic one. Moreover, explicit solutions in \cite{alfonsi2016multivariate} are only given for the special case of an exponential propagator matrix.
 
\paragraph{Our contribution.} In this work we extend the results of \citet{garleanu2016dynamic} and \citet{alfonsi2016multivariate}  in a few  crucial directions which are of relevance to recent empirical studies on cross-impact and also of theoretical interest. Our results provide financial insights on the influence of cross-impact on optimal trading strategies and the interplay between cross-impact and alpha decays.
\begin{itemize} 
\item[\textbf{(i)}] Derivation of explicit solutions: We formulate and solve the optimal portfolio choice problem in continuous time, allowing for a general Volterra propagator matrix $G:[0,T]^2\to\R^{N\times N}$ of nonnegative definite type as well as for general progressively measurable signals. Our optimal strategy is derived explicitly in Theorem \ref{thm:stochastic} in terms of resolvents of the operators involved, and it is implemented in Section \ref{sec-numerics}. Theorem \ref{thm:stochastic} extends the results of \cite{garleanu2016dynamic}, as it allows to solve the problem for a general propagator matrix, which includes the exponential decay kernels used in \cite{garleanu2016dynamic}, as well as the power-law kernels from \citep{benzaquen2017dissecting,mastromatteo2017trading} as specific examples. We also allow for general progressively measurable signals, which substantially generalize the mean-reverting signals used in \cite{garleanu2016dynamic} and the diffusion signals used in follow-up papers. 
Theorem \ref{thm:stochastic} also generalizes the results of \cite{alfonsi2016multivariate} in various directions. First we solve the problem in continuous time, which is compatible with the high-frequency trading timescale. We also include stochastic signals and risk factors,  which are crucial for general portfolio choice problems, and turn them from being matrix-valued and deterministic as in \cite{alfonsi2016multivariate}, into being operator-valued and stochastic. Moreover, we allow for a general Volterra propagator matrix $G(t,s)$ instead of a convolution propagator matrix $G(t)$ as used in \citep{alfonsi2016multivariate,garleanu2016dynamic}, and provide explicit solutions in the general case, in contrast to the special case with an exponential propagator matrix, which was solved in \cite{alfonsi2016multivariate}. As mentioned earlier, due to the generality of the propagator matrix, the portfolio choice problem that we solve is non-Markovian and time inconsistent, hence we introduce new tools, in the form of stochastic forward-backward systems of Fredholm equations of the second kind (see Section \ref{sec-fred}), in order to solve the corresponding first-order conditions. 

\item[\textbf{(ii)}] Preventing price manipulation: In Theorem \ref{thm:convolution} we give sufficient conditions for a convolution propagator matrix $G(t)$ to be nonnegative definite. This means that the expected costs caused by transient price impact are nonnegative for any trading strategy, so that in particular price manipulation in the sense of \citet{huberman2004price} is prevented (see \eqref{round-t}). Specifically, we show that if $G$  is nonincreasing, convex, nonnegative and symmetric, then $G$ is nonnegative definite, and price manipulation is excluded. Our result generalizes Theorem 2  of  \cite{alfonsi2016multivariate} from discrete to continuous time domains, and is of independent interest to the theory of nonnegative definite Volterra kernels. In the single-asset case, nonnegative definiteness is known to hold whenever $G$ is nonnegative, nonincreasing and convex (see Example 2.7 in \citet{GSS}). 
Theorem~\ref{thm:convolution} extends this result to the multi-asset case, and in particular recovers the latter when setting $N=1$ (see Section \ref{subsec:manipulation} for further details).  In Corollary~\ref{cor:product of matrix and real function} we address a popular example from the econophysics  literature (see e.g., \cite{mastromatteo2017trading}). Namely, we prove that if the convolution propagator matrix is factorized as $G(t)= C\phi(t)$, where $C$ is a symmetric nonnegative definite matrix and the function $\phi: [0,T] \rr \mathbb R$ is nonnegative, nonincreasing and convex on $(0,T)$, then $G$ is nonnegative definite, hence price manipulations are excluded.

\item[\textbf{(iii)}] Insights on the influence of cross-impact:
Our numerical study in Section~\ref{sec-numerics} explores the influence of transient cross-impact on the optimal liquidation of Asset 1, with an initial inventory of 0 in Asset 2. The conclusions can be summarized as follows:
\begin{itemize}
    \item[\textbf{(a)}] In the absence of signals, cross-impact induces a transaction-triggered round trip in Asset 2, prompting more aggressive trading in both assets. The strategy is particularly aggressive for the exponential propagator matrix compared to the more persistent fractional one (see Figure \ref{fig:withoutsignalwithcross}).
\item[\textbf{(b)}]
In the presence of positive alpha signals on both assets, cross-impact leverages the different alpha decays. The optimal strategy may involve selling or shorting the asset with the fastest alpha decay to leverage the cross-impact effect that decreases the price of the other asset with the slowest alpha decay. This allows profiting from both the more persistent signal on the second asset as well as the cross-impact effect (see Figures \ref{fig:withsignalcross} and \ref{fig:withsignalcross2}). 
\end{itemize}
\end{itemize} 

\paragraph{Organization of the paper.}  
In Section \ref{res-sec} we present our main results regarding the solution to the portfolio choice problem and the prevention of price manipulation. Section \ref {sec-numerics} is dedicated to numerical illustrations of our main results. In Section \ref{sec-fred} we derive explicit solutions to a class of systems of stochastic Fredholm equations arising from the first-order condition in the proof of Theorem \ref{thm:stochastic}. Finally, Sections \ref{sec-proof-strat}--\ref{sec-proof-lem-ker}
 are dedicated to the proofs of our main results. 

\section{Model Setup and Main Results} \label{res-sec}
\subsection{Model setup} 
Motivated by \cite{garleanu2016dynamic}, we introduce in the following a variant of the optimal portfolio choice problem with transient price impact driven by a Volterra propagator matrix and in the presence of general alpha signals. 

Let $(\Omega, \mathcal{F},\{\mathcal{F}_t\}_{0\leq t\leq T},\P)$ be a filtered probability space satisfying the usual conditions of right-continuity and completeness and let $T>0$ be a deterministic finite time horizon. We consider $N\in\N$ risky assets whose unaffected price vector is given by an $N$-dimensional semimartingale $P=(P_t)_{0\leq t\leq T}=(P_t^1,\dots,P_t^N)^\top_{0\leq t\leq T}$ 
with a canonical decomposition 
\be \label{p-dec}
P=A+M. 
\ee
Here $A=(A_t^1,\dots,A_t^N)^\top_{0\leq t\leq T}$ is a predictable finite-variation process with 
$$
\E\left[\int_0^T \|A_t\|^2dt\right]<\infty, 
$$
and $M=(M_t^1,\dots,M_t^N)^\top_{0\leq t\leq T}$ is a {continuous} martingale  such that 
\be \label{mart} 
d[M^i,M^j]_t=\Sigma_{ij}dt,\quad\text{for }i,j=1,\dots, N,
\ee 
where $\Sigma=(\Sigma_{ij})\in\R^{N\times N}$ is a symmetric nonnegative definite covariance matrix and $\|\cdot\|$ denotes the Euclidean norm.

We consider an investor whose initial portfolio is given by $X_0\in\R^N$, where the amounts of shares held in each of the $N$ risky assets are given by
\be \label{inv}
X_t^u =X_0+\int_0^tu_sds, \quad 0\leq t\leq T. 
\ee
Here $(u_s)_{0\leq s\leq T}=(u_s^1,\dots,u_s^N)^\top_{0\leq s\leq T}$ denotes the trading speed chosen by the investor from the set of admissible strategies
    \be \label{admis}
    \mathcal{U} \vcentcolon= \Big\{u:\Omega\times [0,T]\to \R^N \Big|u \text{ progressively measurable with }\E\Big[\int_0^T \|u_t\|^2dt\Big]<\infty \Big\}.
    \ee
We assume that the investor's trading activity causes linear temporary impact on the assets' execution prices given by 
\be \label{temp}
\frac{1}{2}\Lambda u_t, \quad 0\leq t\leq T, 
\ee
where $\Lambda\in\R^{N\times N}$ is a positive definite (not necessarily symmetric) matrix, i.e. $x^\top\Lambda x>0$ for all $x\in\R^N$. Note that the diagonal entries of $\Lambda $ capture the temporary self-impact while the off-diagonal entries $\Lambda_{ij}$ introduce temporary cross-impact on the price of asset $i$ caused by trading one share of asset $j$ per unit of time.
\begin{remark}
Note that empirical findings in \citet{capponi2020multi}, \citet{coz2023cross} and \citet{cont2023cross} suggest that there is no significant temporary cross-impact of assets, i.e.~that the off-diagonal entries of $\Lambda$ are zero. Incorporating these findings into our model would turn $\Lambda$ into a diagonal matrix with positive entries that capture the temporary self-impact of the assets \`a la \citet{almgren2001optimal}.
\end{remark}
The investor's trading activity also creates a price distortion which is given by 
\be \label{dist} 
D_t^u\vcentcolon=\int_0^t G(t,s)u_sds,\quad 0\leq  t \leq T, 
\ee
where $G:[0,T]^2\to\R^{N\times N}$ is a matrix of Volterra kernels, i.e.~each entry in the matrix satisfies $G_{ij}(t,s)=0$ for $t< s$. The so-called propagator matrix $G$ captures linear transient self-impact in the diagonal terms along with linear cross-impact in the off-diagonal terms. 
That is, the value $G_{ij}(t,s)$ describes the impact at time $t$ on the $i$-th asset's price caused by trading at time $s$ one share of the $j$-th asset per unit of time.
We further say that $G$ is nonnegative definite, if it holds for every $f\in L^2([0,T],\mathbb R^N)$ that 
\be \label{non-neg}
 \int_{0}^T \int_{0}^T f(t)^\top G(t,s)f(s)dsdt \geq 0. 
\ee 
We define $L^2([0,T]^2,\R^{N\times N})$ to be the space of Borel-measurable matrix-valued kernels $G:[0,T]^2\to\R^{N\times N}$ satisfying 
\be \label{l-2-ker-def}
\int_0^T\int_0^T \|G(t,s)\|^2dsdt <\infty.
\ee
Here, $\|\cdot\|$ denotes the Frobenius norm in consistency with our notation $\|\cdot\|$ for the Euclidean norm. 
The set of admissible kernels $\mathcal{G}$ is given by, 
\be \label{prop-class}
\mathcal{G}\vcentcolon=\Big\{G\in L^2([0,T]^2,\R^{N\times N})\Big| \,G\text{ is nonnegative definite with }G(t,s)=0\text{ for }t< s\Big\}.
\ee
\begin{example}\label{ex:propagators}
We present some typical examples for propagator matrices which arise from applications and can be incorporated into our model. 
\begin{enumerate}
\item[\textbf{(i)}] Exponential decay: Let $C$ be a symmetric nonnegative definite matrix in $\R^{N\times N}$. Motivated by \citet{obizhaeva2013optimal}, \citet{alfonsi2016multivariate} investigate the matrix exponential kernel 
    \be G(t,s)=\mathds{1}_{\{t\geq s\}}\exp(-(t-s)C).\ee 
    \citet{garleanu2016dynamic} employ the factorized exponential kernel 
    \be G(t,s)=\mathds{1}_{\{t\geq s\}}e^{-R(t-s)}C \ee 
with scalar price resiliency $R>0$.
\item[\textbf{(ii)}] Power-law decay: \citet{benzaquen2017dissecting} and \citet{mastromatteo2017trading} use for their empirical studies the factorized power-law kernel
    \be G(t,s)=\mathds{1}_{\{t\geq s\}}\big(1+\frac{t-s}{t_0}\big)^{-\beta}C,
\ee
where $C\in\R^{N\times N}$ is symmetric nonnegative definite, $\beta\in(0,1)$ and $t_0>0$.  If $\beta<\frac{1}{2}$, our model also allows the factorized kernel
\be 
G(t,s)=\mathds{1}_{\{t>s\}}(t-s)^{-\beta}C,
\ee
with singular power-law decay as introduced by \citet{gatheral2010no}.
\item[\textbf{(iii)}] Permanent impact: The constant kernel 
\be G(t,s)=\mathds{1}_{\{t\geq s\}}C,
\ee
for a symmetric nonnegative definite matrix $C\in\R^{N\times N}$ incorporates permanent self- and cross-impact as discussed in Section 5 of \citet{huberman2004price}, Section 3 of \citet{schneider2019cross} and Example 1 of \citet{alfonsi2016multivariate}. If $C$ is a diagonal matrix with nonnegative entries, $G$ represents permanent impact \`a la \citet{almgren2001optimal}.
\item[\textbf{(iv)}] Constructed decay kernels: Motivated by Section 3.1 of \cite{alfonsi2016multivariate}, the general kernel 
\be
G(t,s)=Q^T \mathds{1}_{\{t\geq s\}}\operatorname{diag}\big(g_1(t-s),\dots,g_N(t-s)\big)Q
\ee
for an invertible matrix $Q\in\R^{N\times N}$ and nonnegative, nonincreasing, convex kernel functions $g_1,\dots,g_N:[0,T]\to\R$, is nonnegative definite.  
\item[\textbf{(v)}] Interest rate derivatives: The following non-convolution kernel is a straightforward generalization of the price impact for bonds trading model proposed in Section 3.1 of  \citet{brigo2020}, to a portfolio of bonds,
 $$
G(t,s) =  \alpha (T-t) \mathds 1_{\{t>s\}}H(t-s)  C.
$$
Here $H:[0,T] \rr \mathbb{R}$ is a nonnegative, nonincreasing, convex function taking values in $\mathbb{R}$, such as the exponential or power-law kernels in the above examples and $C$ is a symmetric nonnegative definite matrix $C\in\R^{N\times N}$. The factor $ \alpha(T-t)$ for $\alpha \in (0,\infty)$ is added in order to enforce a terminal condition on the bond price regarding its expiration at time $T$.   
\end{enumerate}
\end{example}

\begin{lemma}\label{lemma:examples}
All propagator matrices introduced in Example \ref{ex:propagators} belong to the class of admissible kernels $\mathcal{G}$.
\end{lemma}
Lemma \ref{lemma:examples} is proved in Section \ref{sec-proof-lem-ker}. 

The investor's objective is the maximization of the following portfolio performance criterion over the time period $[0,T]$, 
\begin{align}\label{eq:Ju}\hspace{-5mm}
J(u)\vcentcolon = \mathbb{E}\left[\int_0^T -u_t^\top\big(P_t + D_t^u +\frac{1}{2}\Lambda u_t\big)dt+(X_T^u)^\top P_T-\frac{\gamma}{2}\int_0^T(X_t^u)^\top\Sigma X_t^udt \right]
\end{align}
over all admissible trading strategies $u \in \mathcal{U}$. The first two terms in the right-hand side of \eqref{eq:Ju} describe the investor’s terminal wealth, in terms of her final cash position resulting from trading the assets in the presence of price impact as prescribed above, as well as her remaining final portfolio’s book value. The last term in \eqref{eq:Ju} represents the portfolio's risk in the sense of \citet{Markowitz:1952aa} (see also Section 1 of \cite{garleanu2016dynamic}), where $\Sigma$ was defined in \eqref{mart} and $\gamma\geq 0$ is a risk aversion parameter. 

\begin{remark}\label{rem:distorted value}
One may consider to implement in \eqref{eq:Ju} the final portfolio's distorted value $(X_T^u)^\top(P_T+D_T^u)$ instead of its final book value $(X_T^u)^\top P_T$. Recall the decomposition $P= M+A$. Hence in this case if $A_t$ and $D_t^u$ are both differentiable with respect to $t$, an application of integration by parts yields that maximizing $J(u)$ is equivalent to maximizing the following cost functional, 
\be
\tilde J(u) :=\mathbb{E}\left [\int_0^T \Big( (X_t^u)^\top(\dot{A}_t + \dot{D}_t^u )-\frac{1}{2}u_t^\top\Lambda u_t-\frac{\gamma}{2}(X_t^u)^\top\Sigma X_t^u\Big)dt \right],
\ee
which is a finite time horizon version of the objective functional in \citet{garleanu2016dynamic}. Note however that in this case $\tilde J(\cdot)$ might not be concave for general Volterra kernels $G$, unlike for the exponential kernel case studied in \cite{garleanu2016dynamic}. 
\end{remark}
\subsection{Solution of the portfolio choice problem}
Before we present our results regarding the solution to the portfolio choice problem we introduce some essential definitions and notation. 

We denote the inner product on $L^2([0,T],\R^N)$ by $\langle\cdot,\cdot\rangle_{L^2}$, i.e.
\be
\langle f,g\rangle_{L^2}:=\int_0^T f(t)^\top g(t)dt, \quad f,g\in L^2([0,T],\R^N), 
\ee
and its induced norm by $ \|\cdot \|_{L^2}$.
For any $G\in L^2([0,T]^2,\R^{N\times N})$ (recall \eqref{l-2-ker-def}) define the linear integral operator $\mathbf{G}$ on $L^2([0,T],\R^N)$ induced by $G$ as
\be
(\mathbf{G}f)(t):=\int_0^T G(t,s)f(s)ds,\quad 0\leq t \leq T, \quad f\in L^2([0,T],\R^N). 
\ee
Then $\mathbf{G}$ is a bounded linear operator from $L^2([0,T],\R^N)$ into itself (see Theorem 9.2.4 and Proposition 9.2.7 (iii) in \cite{gripenberg1990volterra}). Moreover, we denote by $G^*$ the adjoint kernel of $G$ with respect to $\langle\cdot,\cdot\rangle_{L^2}$ given by 
\be
G^*(t,s)\vcentcolon = G(s,t)^\top,\quad (s,t)\in [0,T]^2,
\ee
and by $\mathbf{G}^*$ the induced adjoint bounded linear operator on $L^2([0,T],\R^N)$. 

\begin{remark}\label{rem:pos semidefinite}
Note that if $G\in L^2([0,T]^2,\R^{N\times N})$ then it holds that
\be
\langle f,\mathbf{G}f\rangle_{L^2}=\langle f,\mathbf{G}^*f\rangle_{L^2}=\frac{1}{2}\langle f,(\mathbf{G}+\mathbf{G}^*)f\rangle_{L^2}, \quad  \text{for all } f\in L^2([0,T],\R^N).
\ee
Moreover, recalling \eqref{non-neg}, the following three statements are equivalent:
\begin{enumerate}
\item [\textbf{(i)}] The integral operator $\mathbf{G}$ is nonnegative definite.
\item [\textbf{(ii)}] The integral operator $\mathbf{G}^*$ is nonnegative definite.
\item [\textbf{(iii)}] The integral operator $\mathbf{G}+\mathbf{G}^*$ is nonnegative definite.
\end{enumerate}
\end{remark}

\begin{definition}\label{def:column-wise}
Let $h\in L^2([0,T],\R^{N\times N})$ denote a matrix-valued function in one variable and $\mathbf{A}$ be a linear operator from $L^2([0,T],\R^N)$ into itself. Then we define the column-wise application of $\mathbf{A}$ to $h$ as
\be
\mathbf{A}\diamond h :[0,T]\to \R^{N\times N},\quad 
\big(\mathbf{A}\diamond h\big)(r)\vcentcolon = \big((\mathbf{A}h_{\bullet 1})(r),\dots,(\mathbf{A}h_{\bullet N})(r)\big),
\ee
where 
\be
h_{\bullet i}:[0,T]\to\R^N,\quad h_{\bullet i}(r)\vcentcolon =h(r)e_i,
\ee
is given by the $i$-th column of $h$ for $1\leq i\leq N$. 
\end{definition}
\begin{definition} \label{trunc-ker}
For a Volterra kernel $G\in L^2([0,T]^2,\R^{N\times N})$ and a fixed $t\in[0,T]$ define
\be
G_t(s,r)\vcentcolon=\mathds{1}_{\{r\geq t\}}G(s,r),
\ee
and let $\mathbf{G}_t$ denote the induced integral operator on $L^2([0,T],\R^N)$.
\end{definition}
\textbf{Notation.} 
Recall that $\Lambda$ was defined in \eqref{temp}. We denote by 
\be \label{lam-bar}
\bar{\Lambda}\vcentcolon=\frac{1}{2}(\Lambda+\Lambda^\top)\in\R^{N\times N}
\ee
the symmetric part of $\Lambda$, which is a positive definite matrix again.

Recall that $\Sigma$ and $\gamma$ were introduced in \eqref{eq:Ju}. 
For a Volterra kernel $G$ as before, we introduce the following Volterra kernels $F,K:[0,T]^2\to\R^{N\times N}$:
\be  \label{h-k-ker}
F(t,s)\vcentcolon = \gamma\Sigma\mathds{1}_{\{t>s\}}(T-t),  \quad K(t,s)\vcentcolon = G(t,s)+F(t,s).
\ee
We denote by $\E_t[\, \cdot\, ]$ the conditional expectation with respect to $\mathcal{F}_t$.

Finally, we define the following linear operator,
\be \label{D-def}
\mathbf{D}\vcentcolon=\mathbf{G}+\mathbf{G}^*+\mathbf{F}+\mathbf{F}^*+\bar{\Lambda}\mathbf{I},
\ee
where $\mathbf{I}$ is the identity operator on $L^2([0,T],\R^N)$, 
and the stochastic process
\be \label{g-def}
g_t\vcentcolon =\E_t[P_T-P_t]-\gamma (T-t)\Sigma X_0, \quad 0\leq t \leq T, 
\ee
where $X_0$ is given in \eqref{inv}. 

Recall that the class of admissible propagator matrices $\mathcal{G}$ was defined in \eqref{prop-class}. In the following theorem we derive the unique maximizer of the portfolio's revenue-risk functional \eqref{eq:Ju} in terms of resolvents. {Recall that the resolvent of a Volterra operator $\mathbf{B}$ is given by $\mathbf{R}^B=\mathbf{I}-(\mathbf{I}+\mathbf{B})^{-1}$.}

\begin{theorem}\label{thm:stochastic}
Let $G \in \mathcal G$. Then the unique maximizer $u^* \in \mathcal U$ of the objective functional $J(u)$ in \eqref{eq:Ju} is given by 
\be
u_t^*=\big((\mathbf{I}+\mathbf{B})^{-1}a\big)(t),\quad 0\leq t \leq T,
\ee
where
\be
a_t\vcentcolon=\bar{\Lambda}^{-1}\Big(g_t-\int_t^TK(r,t)^\top \big(\mathbf{D}_t^{-1}\mathds{1}_{\{t\leq\cdot\}}\E_t[g_\cdot]\big)(r)dr\Big),
\ee
\be
B(t,s)\vcentcolon=-\bar{\Lambda}^{-1}\Big(\mathds{1}_{\{s\leq t\}}\int_t^T K(r,t)^\top \big(\mathbf{D}_t^{-1}\diamond\mathds{1}_{\{t\leq \cdot\}}K(\cdot,s)\big)(r)dr-K(t,s)\Big),
\ee
\be
\mathbf{D}_t\vcentcolon=\bar{\Lambda}\mathbf{I}+\mathbf{K}_t+\mathbf{K}_t^*,
\ee
where $K$, $\mathbf{D}$ and $g$ are defined in \eqref{h-k-ker}, \eqref{D-def} and \eqref{g-def} respectively, and $\mathbf{B}$ is the integral operator induced by the kernel $B$.
\end{theorem}

In the following corollary we derive the maximizer of \eqref{eq:Ju} for the case of a deterministic signal $A$ in \eqref{p-dec}. This assumption considerably simplifies the optimal portfolio choice and is of relevance for practical applications.  

\begin{corollary}\label{thm:deterministic}
Let $G \in \mathcal G$ and assume that $A$ in \eqref{p-dec} is deterministic. Then, the unique maximizer $u^*\in \mathcal U$ of the objective functional $J(u)$ in \eqref{eq:Ju} is given by 
\be
u_t^*=(\mathbf{D}^{-1}g)(t),\quad\quad 0\leq t \leq T,
\ee
where $\mathbf{D}$ is defined in \eqref{D-def} and $g$ in \eqref{g-def} becomes the deterministic function in $L^2([0,T],\R^N)$ given by
\be
g_t=A_T-A_t-\gamma (T-t)\Sigma X_0. \label{eq:gt}
\ee
\end{corollary}
The proofs of Theorem \ref{thm:stochastic} and Corollary \ref{thm:deterministic} are given in Section \ref{sec-proof-strat}. 

In the following remarks we give additional details regarding the contribution and implementation of our results.

\begin{remark}\label{rem:terminal penalty}
If the investor's objective is to optimally liquidate the portfolio, an additional terminal penalty on the remaining inventory can be added to \eqref{eq:Ju}. Specifically, we define the revenue-risk functional as follows:
\be
 \bar J(u) = \mathbb{E} \bigg [\int_0^T \hspace{-2mm}-u_t^\top(P_t + D_t^u +\frac{1}{2}\Lambda u_t)dt+(X_T^u)^\top P_T 
-\frac{\gamma}{2}\int_0^T(X_t^u)^\top\Sigma X_t^udt-\frac{\varrho}{2} (X_T^u)^\top\Pi X_T^u\bigg],
 \ee
where in the last term of $\bar J(u)$, $\varrho$ is a positive constant and $\Pi\in\R^{N\times N}$ is a symmetric nonnegative definite matrix. 
In this case the result of Theorem \ref{thm:stochastic} applies with a slight modification by setting:
\be
F(t,s)\vcentcolon=\mathds{1}_{\{t> s\}}(\gamma(T-t)\Sigma+\varrho\Pi),
\ee
and  
\be
g_t\vcentcolon=\E_t[P_T-P_t]-(\gamma (T-t)\Sigma+\varrho\Pi) X_0. \label{eq:gtwithvarrho}
\ee
\end{remark}

\begin{remark}\label{rmk:numerical}
In Section \ref{sec-numerics} we provide  a numerical implementation for the optimal strategy which is derived in Theorem \ref{thm:stochastic} and Corollary~\ref{thm:deterministic}. Our implementation includes examples of optimal portfolio choice and of optimal liquidation for various choices of signals and propagator matrices. 
\end{remark} 

\begin{remark}\label{rmk:contribution1}
Theorem \ref{thm:stochastic} extends the results of Proposition 2 in \cite{garleanu2016dynamic}, as it allows to solve the problem for general propagator matrices as in \eqref{prop-class}, which include the exponential decay kernels used in \cite{garleanu2016dynamic}, as well as the power-law kernels from \citep{benzaquen2017dissecting,mastromatteo2017trading} and other specific examples outlined earlier in this section. We also allow for general semimartingale signals, which substantially generalize the integrated Ornstein-Uhlenbeck signals used in \cite{garleanu2016dynamic}. 
Theorem \ref{thm:stochastic} also generalizes the results of \cite{alfonsi2016multivariate} in the following directions: First, we solve the problem in continuous time, which is compatible with the high-frequency trading timescale. We also include stochastic signals and risk factors in the objective functional \eqref{eq:Ju}, which are crucial for portfolio choice problems. This turns the portfolio choice problem from being matrix-valued and deterministic as in \cite{alfonsi2016multivariate} into an operator-valued stochastic control problem. We also allow for a general Volterra propagator matrix $G(t,s)$  instead of the convolution propagator matrix $G(t)$ in \cite{alfonsi2016multivariate}, and provide explicit solutions, in contrast to the first-order condition derived in Theorem 3 of \cite{alfonsi2016multivariate}, which is solely solved for exponential propagator matrices (see Example 2 therein). In particular, in the proof of Theorem \ref{thm:stochastic}, we characterize the first-order condition corresponding to the objective functional \eqref{eq:Ju} in terms of a system of coupled stochastic Fredholm equations of the second kind with both forward and backward components. In Section \ref{sec-fred} we develop a method for solving this system explicitly, which is a central ingredient for the proof of Theorem \ref{thm:stochastic}. 
 \end{remark}

\begin{remark}\label{rem:generalization of class of admissible kernels}
In the single-asset case where $N=1$, Theorem~\ref{thm:stochastic}  generalizes the main results of both \cite{abijaber2022optimal} and \cite{abijaber2023equilibrium} to a larger class of admissible kernels $\mathcal{G}$. Indeed, the proof of Theorem \ref{thm:stochastic} relies heavily on the theory of kernels of type $L^2$ and their resolvents. See in particular the proofs of  Proposition \ref{prop:Fredholm} and Lemma \ref{lemma:integrability}. The method of the proof allows us to extend the class of admissible kernels $\mathcal{G}$ such that only square-integrability of the kernels is needed (see \eqref{prop-class} and \eqref{l-2-ker-def}), instead of the more restrictive assumption that
\be
\sup\limits_{t\leq T}\int_0^T |G(t,s)|^2ds + \sup\limits_{s\leq T}\int_0^T |G(t,s)|^2dt <\infty,
\ee
which was assumed by \citet{abijaber2022optimal} and \citet{abijaber2023equilibrium} who studied the single-asset case. In particular our new approach can be applied to prove Proposition 4.5 and Lemma 7.1 in \cite{abijaber2022optimal} as well as Proposition 5.1 and Lemma 5.5 in \cite{abijaber2023equilibrium} under the more general integrability assumption on the admissible kernels $G$ in $\mathcal{G}$ given by
\be
\int_0^T\int_0^T |G(t,s)|^2dsdt<\infty.
\ee
\end{remark}

\subsection{Results on the absence of price manipulation}\label{subsec:manipulation}

In this section we derive conditions on the kernel $G\in L^2([0,T]^2,\R^{N\times N})$ that prevent price manipulation.  
For this, given a strategy $u \in \mathcal U$, we define the associated costs caused by transient impact as
\begin{align}\label{eq:costuGu}
\mathcal{C}(u) :=\int_0^T\int_0^T u(t)^\top G(t,s) u(s)dsdt. 
\end{align}
Note that the performance functional $J(u)$ in \eqref{eq:Ju} reduces to $-\mathcal{C}(u)$ when setting the risk aversion $\gamma$ and temporary price impact matrix $\Lambda$ therein to $0$, choosing a martingale unaffected price process $P=M$ and optimizing over strategies with a fuel constraint ($X_T =0$). We refer to Section 2.1 of \cite{lehalle2019incorporating} for the derivation. 
\begin{remark} Definition \eqref{eq:costuGu} aligns with the definition of  transient impact costs in Section 2.5 of \cite{alfonsi2016multivariate}, which is itself an extension of definition (2.4) in \cite{GSS} to the multi-asset case. In contrast to \cite{alfonsi2016multivariate,GSS} - where admissible inventories $(X_t)_{t\in [0,T]}$ are taken to be left-continuous, adapted, bounded processes whose components have finite and $\P$-a.s.~bounded total variation - we restrict our attention to the subclass of absolutely continuous inventories as defined in \eqref{inv}, so that $dX_t = u(t)dt$. While the framework of \cite{alfonsi2016multivariate,GSS} allows for more general inventory processes, which in particular can have jumps, the special case of absolutely continuous inventories is more tractable and aids in deriving the explicit solution of the portfolio choice problem (see Theorem \ref{thm:stochastic}).
\end{remark}
A propagator matrix $G$ precludes price manipulation in the sense of \citet{huberman2004price}, if $\E\left[\mathcal{C}(u)\right] \geq 0$ for all strategies $u\in\mathcal{U}$ with $\int_0^Tu_tdt=0$ (so-called round trip strategies). This means, in the absence of signals, that any portfolio starting and terminating with zero inventory $(X_0=X_T=0)$ cannot create any profit, or equivalently, negative transient impact costs.
We observe that the expected transient impact costs $\E\left[\mathcal{C}(u)\right]$ from \eqref{eq:costuGu} are nonnegative for all strategies $u\in\mathcal{U}$ if and only if $G$ is nonnegative definite as in \eqref{non-neg}.
Hence, the condition that $G$ is contained in the class $\mathcal{G}$ from \eqref{prop-class} (and thus nonnegative definite) rules out the possibility of price manipulation in the sense of \cite{huberman2004price}. 
Interestingly enough, the nonnegative definiteness of $G$ is also needed in order to prove that the objective functional \eqref{eq:Ju} is concave (see Lemma \ref{lemma:concavity}), hence it is essential for the proof of Theorem \ref{thm:stochastic}. We refer to Section 2 of \cite{alfonsi2016multivariate} for the study of price manipulation in discrete time. 

Note that the  nonnegative definiteness condition \eqref{non-neg} is not straightforward to verify. For discrete-time formulations and the continuous-time formulation from \mbox{\cite{alfonsi2016multivariate,GSS}} (which permits jumps in the inventory processes), Bochner’s theorem \cite{bochner1932vorlesungen} provides a complete characterization of nonnegative definiteness for convolution kernels $G(t,s): =\mathds{1}_{\{t\geq s\}}H(t-s)$ with $H:[0,T]\to\R^{N\times N}$ via Fourier transforms of Borel measures (see Theorem 1 in \cite{alfonsi2016multivariate} and Proposition 2.6 in \cite{GSS}). However, because we restrict our strategies to be absolutely continuous, we are outside the usual setting of Bochner’s theorem, which relies on the classical notion of positive definite functions. Moreover, in practice, one often prefers simpler explicit conditions on $H$ that ensure its nonnegative definiteness.  

The main result of this section derives sufficient conditions for a large class of convolution kernels to be nonnegative definite (see Theorem \ref{thm:convolution}). For comparison, we recall that in the single-asset case, convolution kernels of the above form are nonnegative definite as in \eqref{non-neg} whenever $H$ is nonnegative, nonincreasing, and convex (see Example 2.7 in \cite{GSS}). 
The following theorem is an extension of this result to the multi-asset case (under the assumption of absolutely continuous inventory processes). Note that we impose only one further assumption, namely symmetry, and recover exactly the conditions from \cite{GSS} when setting $N=1$.

\begin{theorem}\label{thm:convolution}
Let $H:[0,T]\to\R^{N\times N}$ be a convolution kernel so that the associated Volterra kernel $G(t,s)\vcentcolon =\mathds{1}_{\{t\geq s\}}H(t-s)$ is in $L^2([0,T]^2,\R^{N\times N})$. Assume that $H$ satisfies the following assumptions: 
\begin{itemize}
    \item nonincreasing, i.e.~the function $t\mapsto x^\top H(t) x$ is nonincreasing on $(0,T)$ for every $x\in\R^N$,
    \item convex, i.e.~the function $t\mapsto x^\top H(t) x$ is convex on $(0,T)$ for every $x\in\R^N$,
    \item nonnegative, i.e.~the matrix $H(t)$ is nonnegative definite for any $t\in (0,T)$,
    \item symmetric, i.e.~the matrix  $H(t)$ is symmetric for any $t\in (0,T)$.
\end{itemize}
Then $H$ is nonnegative definite and $G$ is in the class of admissible kernels $\mathcal{G}$, i.e.
\be
\int_0^T\int_0^t f(t)^\top H(t-s)f(s)dsdt=\langle f,\mathbf{G}f\rangle_{L^2} \geq 0, \quad \textrm{for all }f\in L^2([0,T],\R^N).
\ee
\end{theorem}
The proof of Theorem \ref{thm:convolution} is given in Section \ref{sec-pf-mani}.
\begin{remark}
The convexity of $H$ implies that the function $t\mapsto x^\top H(t) x$ is continuous on $(0,T)$ for any $x\in\R^N$. Thus, it follows from the symmetry of $H$ that it is continuous on $(0,T)$. Nevertheless, Theorem \ref{thm:convolution} permits kernels $H$ which have a singularity at $t=0$ such as the power-law propagator in Example \ref{ex:propagators}(ii). 
\end{remark}

\begin{remark} 
Theorem \ref{thm:convolution} generalizes the result of Theorem 2 of  \cite{alfonsi2016multivariate} from discrete time grids to a continuous domain, which is of independent interest to the theory of nonnegative definite Volterra kernels. One of the main ingredients of the proof is Proposition \ref{prop:pos semidef kernels}, which states that a matrix-valued Volterra kernel which is nonnegative definite in the discrete sense (see \eqref{eq:positive defnite kernel}) is also nonnegative definite on a continuous domain, i.e.~it satisfies \eqref{non-neg}. Proposition \ref{prop:pos semidef kernels} then allows us to generalize Theorem 2 from \cite{alfonsi2016multivariate}, which is a discrete-time result proved via methods related to Bochner's theorem, to our continuous-time setting, where Bochner's theorem doesn't  apply directly. Thereby, we can establish sufficient conditions under exactly the same assumptions as in the discrete-time framework. 
\end{remark}

The following corollary rules out the possibility of price manipulation in the sense of \cite{huberman2004price} for a specific class of factorized convolution propagator matrices introduced by \citet{benzaquen2017dissecting} and \citet{mastromatteo2017trading}. 
\begin{corollary}\label{cor:product of matrix and real function}
Let $C\in\R^{N\times N}$ be a symmetric nonnegative definite matrix and $\phi\in L^2([0,T],\R)$ be nonnegative, nonincreasing and convex on $(0,T)$. Then the convolution kernel
\be \label{h-cfm}
H:[0,T]\to\R^{N\times N},\quad H(t)\vcentcolon=C\phi(t)
\ee
is nonnegative definite and thus its associated Volterra kernel is in $\mathcal{G}$.
\end{corollary}
The proof of Corollary \ref{cor:product of matrix and real function} is postponed to Section \ref{sec-pf-mani}.
\begin{remark}\label{R:crosskernelBouchaud}
From Corollary \ref{cor:product of matrix and real function} it follows that our model accommodates the factorized model introduced in \cite{benzaquen2017dissecting, mastromatteo2017trading}, where empirical evidence for the explanatory relevance of cross-impact for the cross-correlation of US stocks was given. The authors found that as a good and efficient approximation the price impact kernel can be written as $G(t,s)=C\phi(t-s)$ for a symmetric nonnegative definite matrix $C\in\R^{N\times N}$ and a power-law function
\be
\phi(t)=(1+t/t_0)^{-\beta} \mathds{1}_{\{t\geq 0\}},
\ee
where $\beta\in(0,1)$ and $t_0>0$ are estimated model parameters. The main conclusions of their papers are also briefly summarized in Chapter~14.5.3 of \cite{bouchaud_bonart_donier_gould_2018}.
\end{remark}

\section{Numerical Illustrations} \label{sec-numerics}

\subsection{Optimal liquidation: what is the influence of cross-impact?}

\noindent\textbf{Two assets.} We first explore the influence of transient cross-impact for two assets in a liquidation problem with time horizon $T=10$, where Asset 1 has $10$ shares initial inventory and Asset 2 starts from zero inventory, i.e.~$X_0=(10,0)^\top$. Our objective is to elucidate the influence of cross-impact on optimal trading speeds and inventories and its interplay with  trading signals (also called alphas). To isolate these effects, we consider a penalty on terminal inventory using the parameters $(\varrho, \Pi)$ of the form
\begin{align}\label{eq:varrho}
\varrho = {4}, \quad \Pi = \left(\begin{matrix} {1} & {0} \\ {0} & {1} \end{matrix}\right),
\end{align}
and we set $\gamma=0$, removing the risk component in the objective functional $\bar J$ in Remark \ref{rem:terminal penalty} originating from the covariance matrix between the assets.

We consider temporary self-impact but no  temporary cross-impact through the  $2\times 2$-matrix
 \be \label{exmp-lam}
  \Lambda = \left(\begin{matrix} {0.03} & 0 \\ 0 & {0.03} \end{matrix}\right).
  \ee
The  transient price impact is given by the parsimonious propagator from \cite{benzaquen2017dissecting,mastromatteo2017trading} (see Corollary \ref{cor:product of matrix and real function}) of the form
    \begin{align}\label{eq:Aphi}
        G(t,s) = C \phi(t-s) \mathds 1_{\{s<t\}},
    \end{align}
      with a deterministic symmetric  $2\times 2$-matrix $C$ that will be taken to be either diagonal or non-diagonal to study the impact of cross-impact (see \eqref{eq:Adiag}-\eqref{eq:Afull} below), and 
       a scalar convolution kernel $\phi$ which will be either set to $\phi^{\text{zero}}(t) = 0$, exponential  $\phi^{\text{exp}}(t)=e^{-\rho t }$ with $\rho = 0.5$ or fractional  $\phi^{\text{frac}}(t)=t^{-\alpha}$ with $\alpha =0.25$.

We analyze four  figures to illustrate our findings:

\begin{itemize}
\item Figure \ref{fig:withoutsignalwithoutcross}: No trading signal, no cross-impact with 
\begin{align}\label{eq:Adiag}
    C = \left(\begin{matrix} {0.06} & {0} \\ {0} & {0.06} \end{matrix}\right).
\end{align}
\item Figure \ref{fig:withoutsignalwithcross}: No trading signal, with cross-impact with
\begin{align}\label{eq:Afull}
    C = \left(\begin{matrix} {0.06} & {0.05} \\ {0.05} & {0.06} \end{matrix}\right).
\end{align}

\item Figures~\ref{fig:withsignalcross}-\ref{fig:withsignalcross2}:   With `noisy'  trading signals of different decay $\beta$ as follows:
\begin{align}\label{eq:signalP}
    dP_t = I_t dt + dM_t,
\end{align}
    with $M$ a martingale and 
    \begin{align}\label{eq:signalI}
         dI_t = -\beta I_tdt + dW_t, 
    \end{align}
    where $W$ is a Brownian motion and 
    \begin{align}\label{eq:beta}
       I_0 = \left(\begin{matrix} {0.5 } \\ {0.5} \end{matrix}\right) \quad  \text{ and } \quad   \beta  =\left(\begin{matrix} \beta_1 & {0} \\ {0} & \beta_2\end{matrix}\right). 
    \end{align}
We consider the cases without and with cross-impact given in \eqref{eq:Adiag} and \eqref{eq:Afull} (respectively).
\end{itemize}

Our conclusions are summarized as follows:
\begin{enumerate}
\item  In the absence of trading signals (see Figures~\ref{fig:withoutsignalwithoutcross} and \ref{fig:withoutsignalwithcross}), cross-impact induces:
  \begin{itemize}
  \item 
 A `transaction-triggered' round trip in Asset 2 on the bottom panels of Figure~\ref{fig:withoutsignalwithcross} which was absent from Figure~\ref{fig:withoutsignalwithoutcross}. Initiating fast liquidation of Asset 1 at time $t=0^+$  triggers a drop in Asset 2's price due to $C_{12}>0$, prompting a `transaction-triggered' shorting signal on Asset 2. Therefore, the optimal strategy starts by shorting Asset 2. After a while, as the selling speed of Asset 1 slows down, the strategy starts buying Asset 2 at a steady rate. This strategic move aims at increasing the price of Asset 1 undergoing liquidation, using the positive cross-impact term $C_{21}>0$. This steady buying even results in a long position in Asset 2 which is then liquidated quickly near the maturity.  {The strategy is more aggressive for the exponential kernel compared to the fractional one.}
 \item A more aggressive trading in Asset 1 even in the presence of transient impact. Inventories with transient impact align closer to those without transient impact on the top panels of Figure~\ref{fig:withoutsignalwithcross} than on those of Figure~\ref{fig:withoutsignalwithoutcross}.
  \end{itemize}
   \item 
     In the presence of positive trading signals on both assets with different alpha decays (see Figures~\ref{fig:withsignalcross} and \ref{fig:withsignalcross2}): 

     \begin{itemize}
         \item \textbf{Without cross-impact (solid lines):} In both figures, the positive signal for Asset 1 prompts from the start a purchasing of Asset 1 even in the presence of transient impact. The `buy' strategy is more aggressive in the absence of transient impact. Similarly, the positive signal of Asset 2 triggers a `buy round trip' on Asset 2.
         \item \textbf{With cross-impact (dashed lines)}: 
               In Figure~\ref{fig:withsignalcross} with $(\beta_1,\beta_2)=(0.9,0.3)$ in \eqref{eq:beta} we observe that contrary to the no cross-impact case, there is almost no buying of Asset 1 at $t=0^+$, even in the presence of the positive signal for Asset 1. Indeed, since the `buy' signal for Asset 1 decays more rapidly compared to that of Asset 2, recall the mean-reversion coefficients $\beta$ in \eqref{eq:beta},  the optimal strategy starts selling Asset 1 quickly (and even shorts it) in order to leverage the cross-impact that will decrease the price of Asset 2 that has a more persistent positive signal. Hence, this strategy allows profiting from both the persistent signal on Asset 2  and from the cross-impact effect. More of Asset 2 is bought compared with the case without cross-impact.  
               
In Figure~\ref{fig:withsignalcross2} with $(\beta_1,\beta_2)=(0.3,0.9)$ in \eqref{eq:beta} we observe the opposite effect, since now Asset 1 has the more persistent signal.
          \end{itemize}
 \end{enumerate}
 
	\begin{figure}[htb]
		\includegraphics[scale=.6,center]{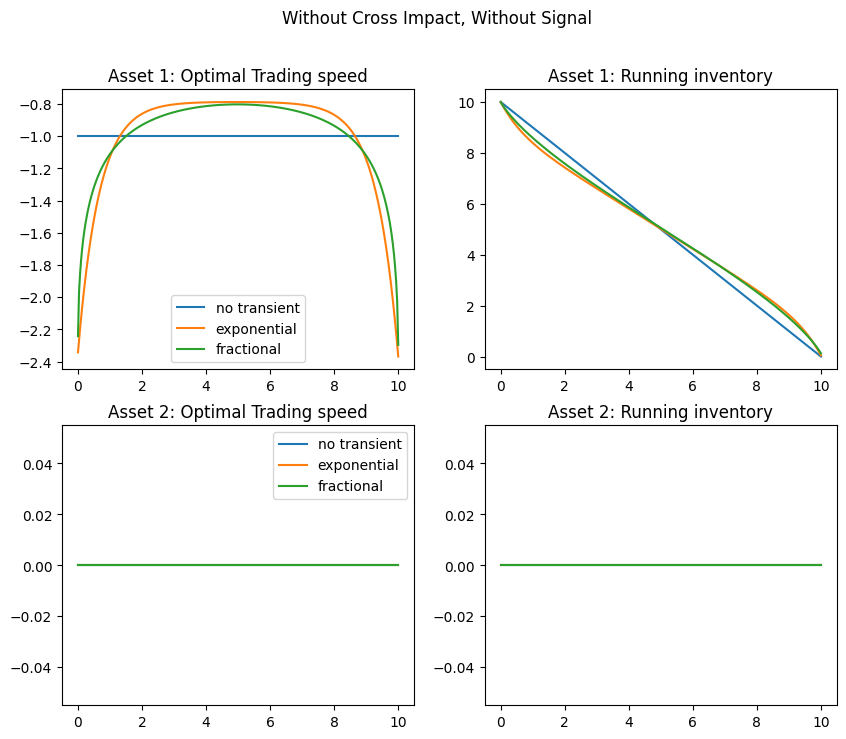}
				\vspace{-1cm}
	\captionof{figure}{Impact of different kernels on the optimal trading speed and inventory of two assets in the absence of a signal \textbf{without cross-impact} \eqref{eq:Adiag} with the three impact kernels $\phi^{\text{zero}}$ (blue), $\phi^{\text{exp}}$  (yellow) and $\phi^{\text{frac}}$ (green). }
\label{fig:withoutsignalwithoutcross}
	\end{figure}

	\begin{figure}[!htb]
		\includegraphics[scale=.63,center]{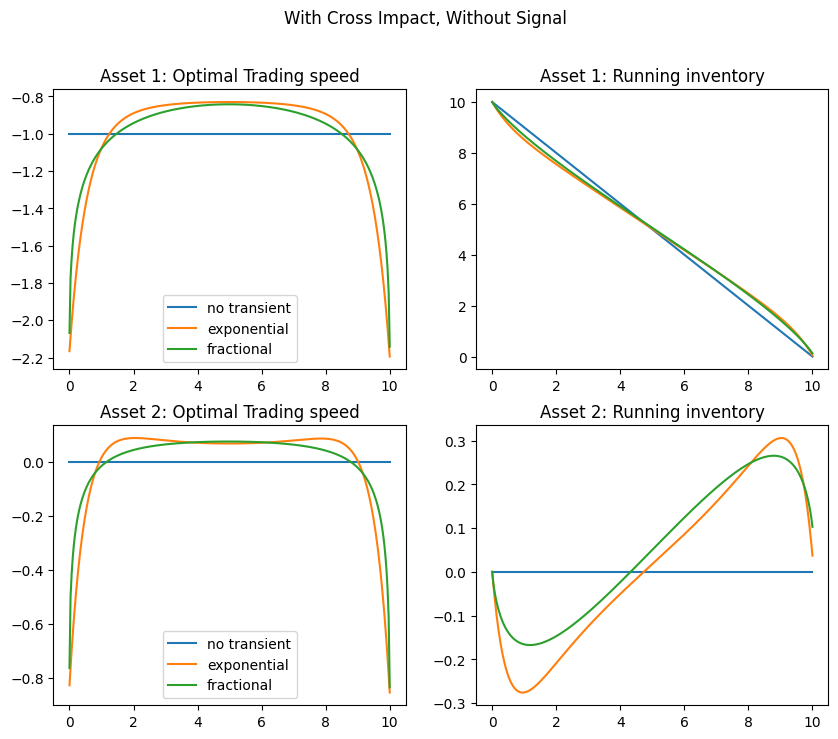}
				\vspace{-1cm}
		\captionof{figure}{Impact of different kernels on the optimal trading speed and inventory of two assets in the absence of a signal \textbf{with cross-impact}  \eqref{eq:Afull} with the three impact kernels $\phi^{\text{zero}}$ (blue), $\phi^{\text{exp}}$  (yellow) and $\phi^{\text{frac}}$ (green). }
\label{fig:withoutsignalwithcross}
	\end{figure}

\begin{figure}[!htb]
		\includegraphics[scale=.6,center]{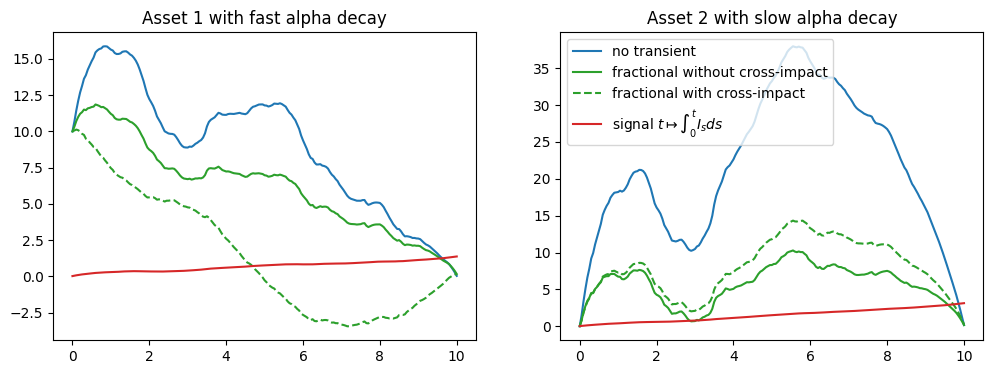}
				\vspace{-1cm}
	\captionof{figure}{Impact of the fractional transient cross-impact  on the optimal inventory of two assets in the presence of buy trading signals with different alpha decay on each asset given by $(\beta_1,\beta_2)=(0.9,0.3)$ in \eqref{eq:beta} \textbf{without (solid) and with (dashed) cross-impact} with the impact kernels $\phi^{\text{zero}}$ (blue) and $\phi^{\text{frac}}$ (green). }
\label{fig:withsignalcross}
	\end{figure}

\begin{figure}[!htb]
		\includegraphics[scale=.6,center]{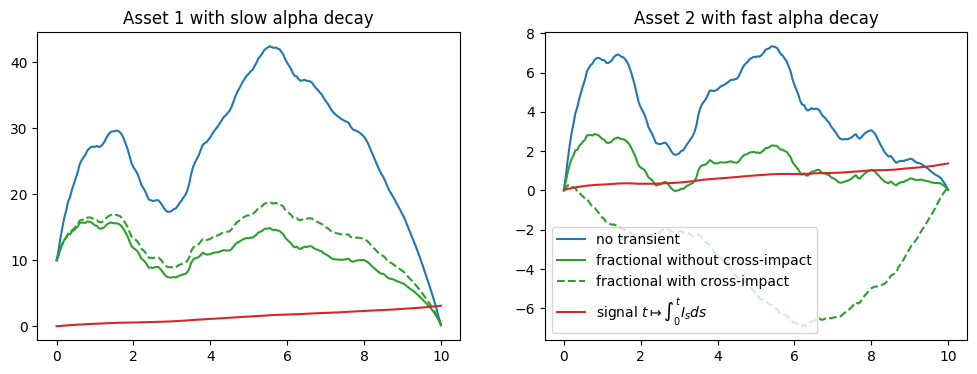}
				\vspace{-1cm}
	\captionof{figure}{Impact of the fractional transient cross-impact  on the optimal inventory of two assets in the presence of buy trading signals with different alpha decay on each asset  given by $(\beta_1,\beta_2)=(0.3,0.9)$ in \eqref{eq:beta} \textbf{without (solid) and with (dashed) cross-impact} with the impact kernels $\phi^{\text{zero}}$ (blue) and $\phi^{\text{frac}}$ (green). }
\label{fig:withsignalcross2} 
	\end{figure}

\textbf{Three assets.} We now extend Figure~\ref{fig:withoutsignalwithcross} to a three-asset setting, keeping the same time horizon and (unless stated otherwise) the same parameter values as above. We consider initial inventory $\smash{X_0 = (10,0,0)^\top}$, and
\begin{align}
\varrho=4,\quad \Pi = \begin{pmatrix}
1 & 0 & 0 \\
0 & 1 & 0 \\
0 & 0 & 1
\end{pmatrix},
\quad \Lambda = \begin{pmatrix}
0.03 & 0 & 0 \\
0 & 0.03 & 0 \\
0 & 0 & 0.03
\end{pmatrix}.
\end{align}
The transient impact again takes the form \eqref{eq:Aphi} with a symmetric $3\times 3$-matrix $C$. To highlight higher-dimensional interaction effects, we choose a `chain cross-impact' configuration given by
\begin{align}\label{eq:Achain-3d}
C = \begin{pmatrix}
0.06 & 0.05 & 0 \\
0.05 & 0.06 & 0.05 \\
0 & 0.05 & 0.06
\end{pmatrix}.
\end{align}
That is, Asset~1 and Asset~2 have cross-impact, as do Asset~2 and Asset~3, while there is no cross-impact between Asset~1 and Asset~3.

Figure~\ref{fig:3d} illustrates the resulting optimal trading speeds and inventories for the different kernels. We observe a `chain effect': Although Asset~3 does not directly interact with Asset~1, a small-magnitude round trip in Asset~3 can contribute to mitigating the cost of the round trip in Asset~2 through the nonzero entries $C_{23}=C_{32}$, which in turn affects the liquidation of Asset~1 through the nonzero entries $C_{12}=C_{21}$. This example illustrates how indirect interactions between assets can influence optimal trading strategies in higher dimensions.

\begin{figure}[!htb]
  \includegraphics[scale=.6,center]{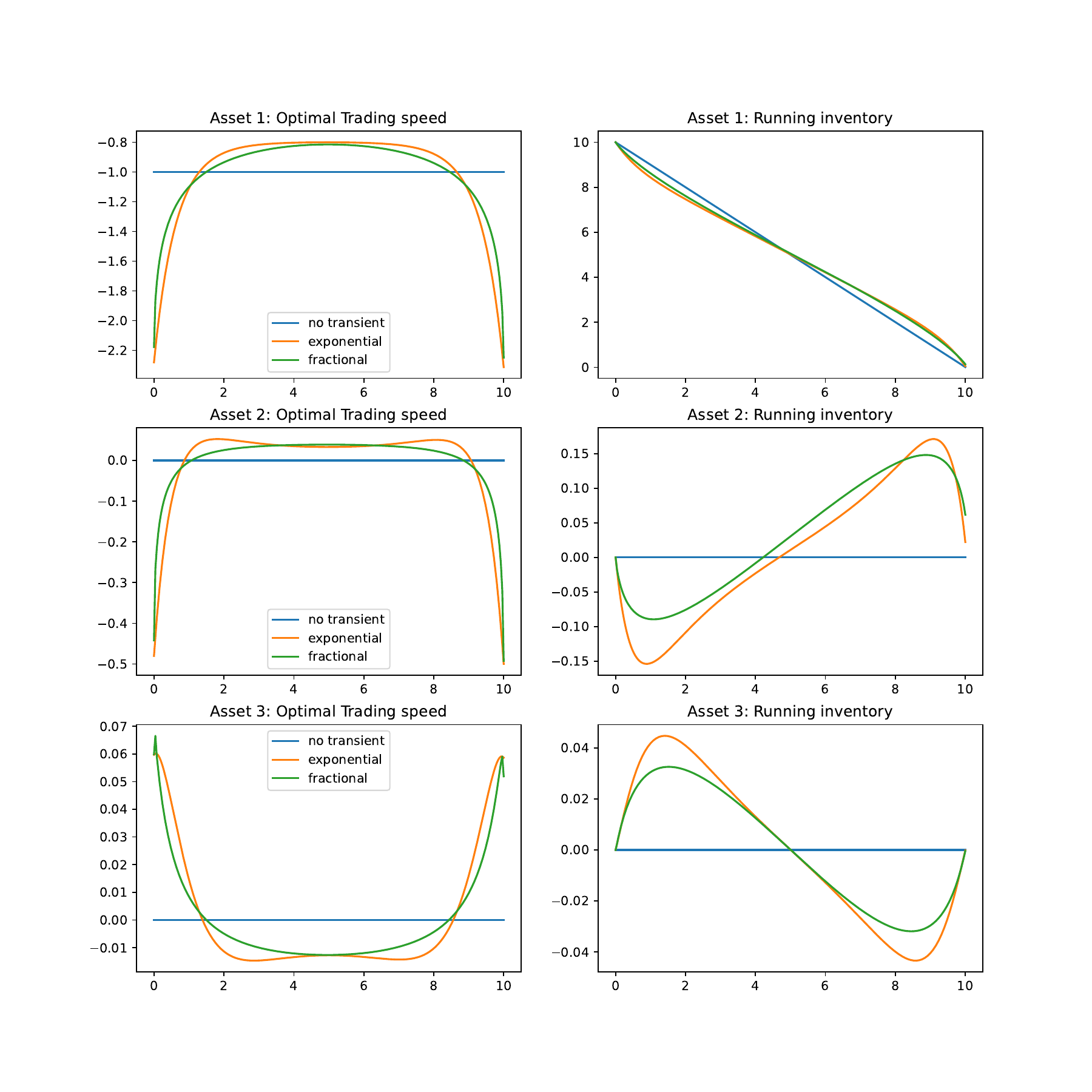}
  \vspace{-1cm}
  \captionof{figure}{Impact of different kernels on the optimal trading speed and inventory of three assets in the absence of a signal \textbf{with chain cross-impact} \eqref{eq:Achain-3d} with the three impact kernels $\phi^{\text{zero}}$ (blue), $\phi^{\text{exp}}$  (yellow) and $\phi^{\text{frac}}$ (green).}
  \label{fig:3d}
\end{figure}

\subsection{Optimal liquidation: what is the influence of correlation?}
In this section, we explore the influence of asset correlation on the optimal liquidation strategy. We consider the same setting as in the previous section, with the following modifications: The cross-impact is set to zero, and the initial position consists of two assets to be liquidated, i.e.~$X_0 = (10, 10)^\top$. 
To differentiate the liquidity of the two assets, we adjust the temporary self-impact matrix to
\begin{equation} 
\Lambda = \begin{pmatrix} 0.03 & 0 \\ 0 & 0.1 \end{pmatrix},
\end{equation}
so that Asset~1 is significantly more liquid than Asset~2.  

We introduce a risk aversion parameter $\gamma = 0.3$, and examine two configurations of the asset correlation matrix:
\begin{itemize}
    \item First row of Figure~\ref{fig:unwind2}: Uncorrelated case with
\be\label{eq:sigma-uncorr}
\Sigma = 0.2\begin{pmatrix} 1 & 0 \\ 0 & 1 \end{pmatrix}.
\ee
\item  Second row of Figure~\ref{fig:unwind2}: Fully correlated case with
\be\label{eq:sigma-corr}
\Sigma =0.2 \begin{pmatrix} 1 & 1 \\ 1 & 1\end{pmatrix}.
\ee
\end{itemize} 
Figure~\ref{fig:unwind2} displays the inventory trajectories in both configurations. In the uncorrelated case (top row), we observe that the introduction of risk aversion leads to a more aggressive initial liquidation, particularly for the more liquid Asset~1. In the fully correlated case (bottom row), the optimal strategy liquidates Asset~1 even more rapidly, temporarily taking a short position in the liquid asset to reduce overall portfolio risk, while unwinding Asset~2 more gradually. This behavior is consistent with the cross-asset execution effect described in~\cite*{almgren2001optimal}.

\begin{figure}[!htb]
	\includegraphics[scale=.6,center]{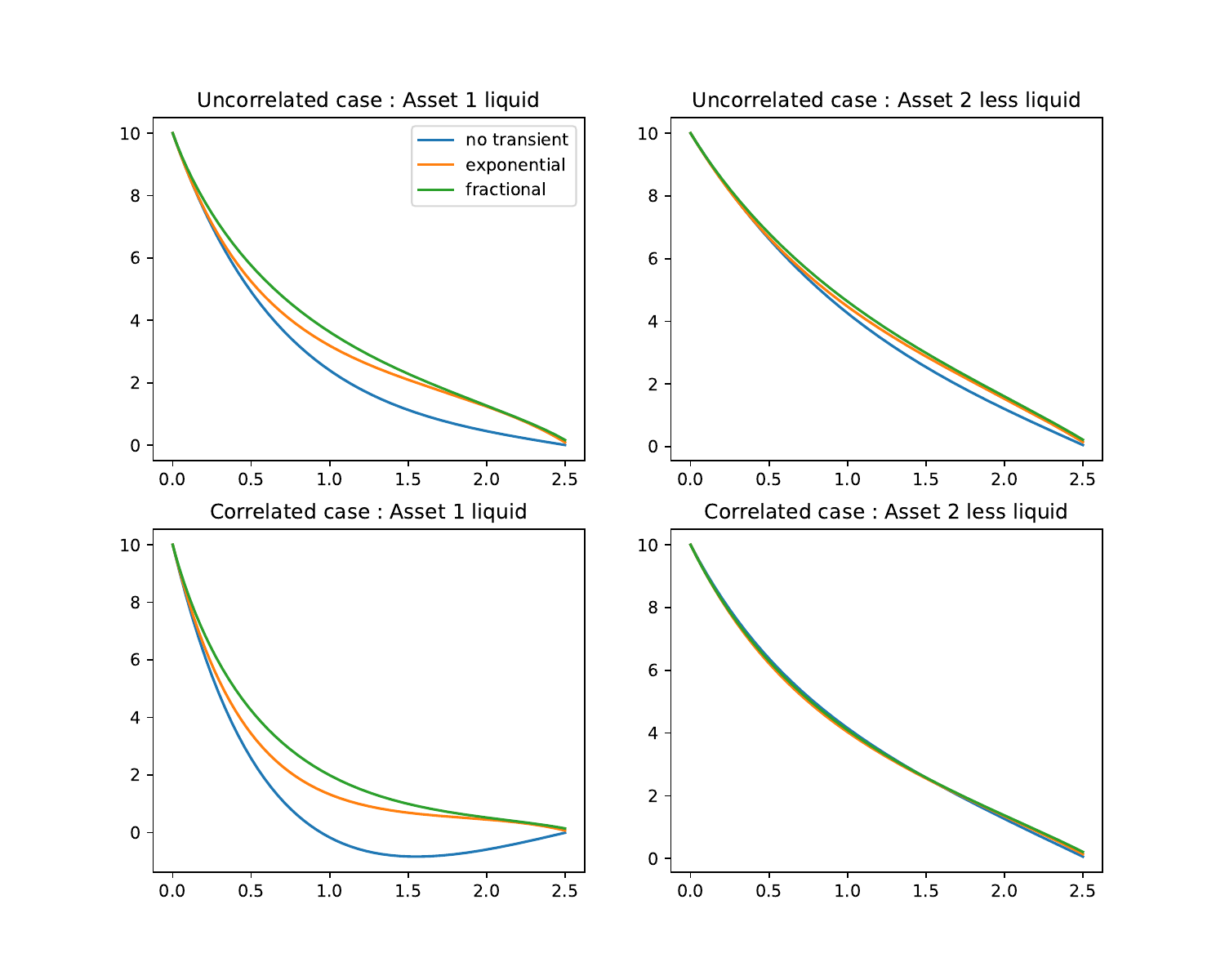}
				\vspace{-1cm}
	\captionof{figure}{Impact of the correlation on the running inventory:  uncorrelated case \eqref{eq:sigma-uncorr} (top row) and fully correlated case \eqref{eq:sigma-corr} (bottom row) with the impact kernels $\phi^{\text{zero}}$ (blue) and $\phi^{\text{frac}}$ (green). }
\label{fig:unwind2} 
\end{figure}

\subsection{Optimal trading with frictions}

In this section, we illustrate the influence of transient impact on a portfolio choice example with two independent assets. 
Here we set  $\varrho =0$ as in the objective functional \eqref{eq:Ju} (i.e. no penalty on terminal inventory) and we set the risk aversion coefficient to $\gamma=5$. We look at the  time evolution of holdings in each asset in comparison with the Markowitz portfolio, i.e. the zero temporary and transient price impact case.
This experiment has already appeared in  \cite{garleanu2016dynamic} for the exponential kernel.

We set $T=10$, $X_0=(7.5,-7.5)^\top$ with
two independent assets with  a $2\times 2$-covariance matrix $\Sigma$ given by   $$  \Sigma = \left(\begin{matrix} 0.04& 0 \\ 0 &0.05  \end{matrix}\right) , $$
and a stochastic  trading signal of the form \eqref{eq:signalP}-\eqref{eq:signalI} with 
     $$ I_0 = \left(\begin{matrix} {0.01 } \\ {-0.01} \end{matrix}\right) \quad  \text{ and } \quad   \beta  =\left(\begin{matrix} {0.05} & {0} \\ {0} & 0.3 \end{matrix}\right).$$
The transient impact is given by \eqref{eq:Aphi}, where we do not consider cross-impact, i.e. $C$ is given by \eqref{eq:Adiag} and $\phi$ is chosen as in the previous section. 

We  make precise the notion of the Markowitz portfolio in this context.  
\begin{remark}
    In the absence of any market frictions, i.e.~$\Lambda =G=0$,
for an unaffected price of the form
$$ dP_t = I_t dt + dM_t,$$
with $M$ a martingale, 
the optimal position is given by the celebrated Markowitz portfolio 
\begin{align}\label{eq:markowitz}
X_t^{\text{Markowitz}} = \frac {\Sigma^{-1} I_t} {\gamma}.
\end{align}
To see this it suffices to set $\Lambda = G = 0$ in the functional $J$ in \eqref{eq:Ju} and apply an integration by parts to get:
$$ \mathbb E\left[  P_T^\top X_T - \int_0^T P_t^\top  u_t dt \right] = X_0^\top  P_0 + \mathbb E\left[ \int_0^T I_t^\top X_t  dt\right], $$
which would then yield the functional
\begin{align}
X_0^\top P_0 + \mathbb E\left[ \int_0^T \left( I_t^\top  X_t -\frac{\gamma}{2} X_t^\top\Sigma X_t \right) dt\right]
\end{align}
and the optimality of the Markowitz portfolio \eqref{eq:markowitz}. \end{remark}

The optimal positions in Assets 1 and 2 without cross-impact are illustrated in  Figure \ref{fig:markowitz} for different values of the temporary impact parameters $\Lambda$.   The Markowitz portfolio is just a reference portfolio that describes the weights of the optimal portfolio derived from the optimal trade-off between risk $\Sigma$ and expected excess return $I_t$ scaled by $\gamma$. We make the following observations:
\begin{itemize}
    \item 
Similarly to  \citet{garleanu2016dynamic}, we observe that if one chooses any initial starting portfolio $X_0$, all optimal positions in the presence of frictions try to follow and get closer to the Markowitz portfolio in a smooth way. 
\item The `speed of convergence' towards the Markowitz portfolio seems to depend on the type of the frictions, being faster for the no-transient kernel than for the exponential kernel, and slowest for the more persistent fractional case. 
\end{itemize}

\begin{figure}[!ht]
\includegraphics[scale=.7,center]{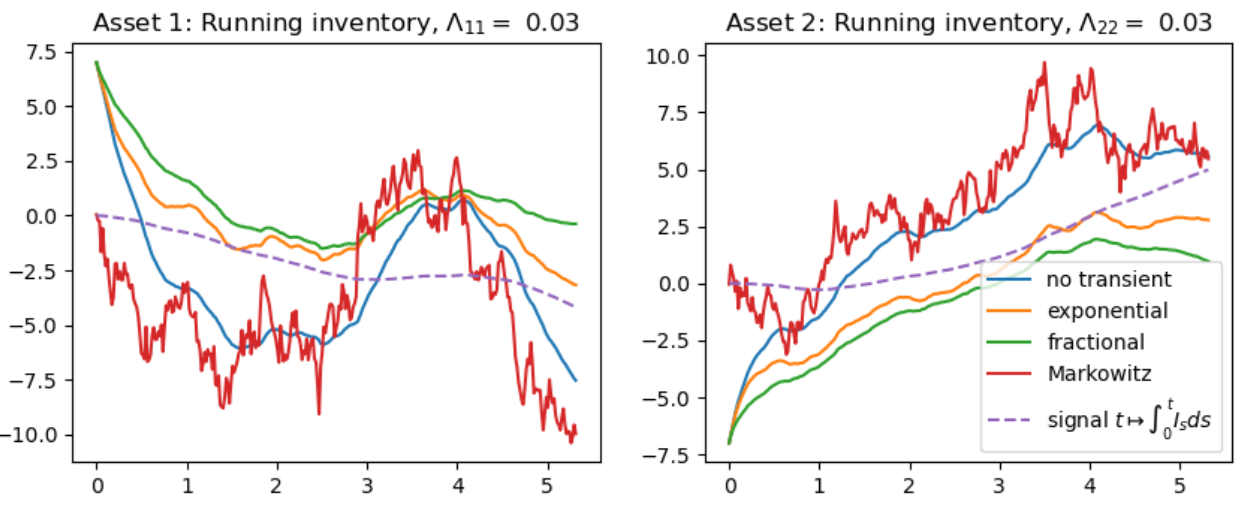}\\
\includegraphics[scale=.7,center]{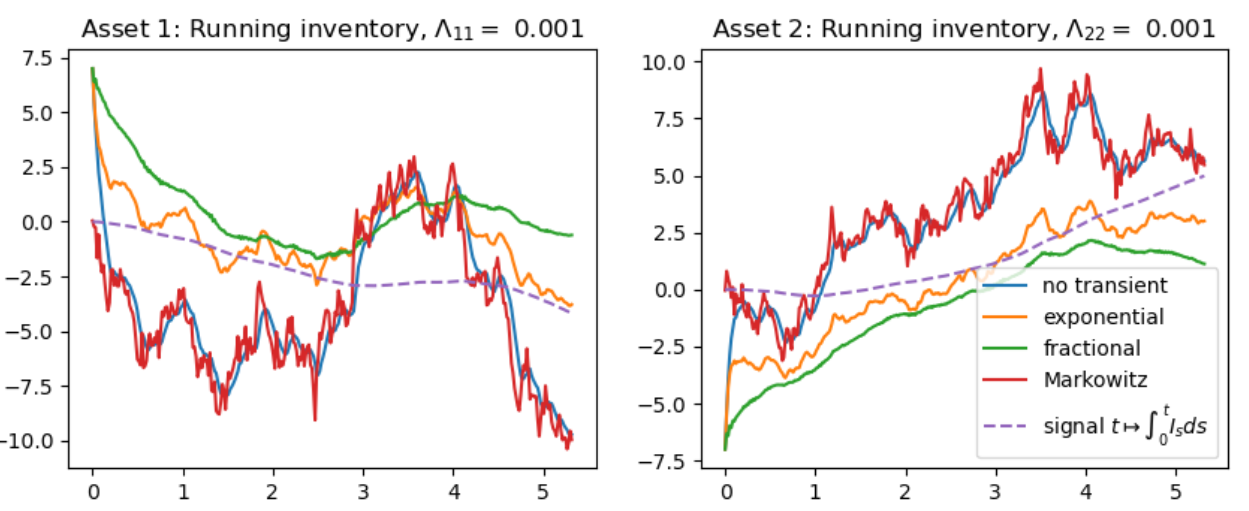}\\
				\vspace{-0.5cm}
	\captionof{figure}{Optimal trading without cross-impact with two independent assets and stochastic signals with the three impact kernels $\phi^{\text{zero}}$ (blue), $\phi^{\text{exp}}$  (yellow) and $\phi^{\text{frac}}$ (green),  and the Markowitz portfolio (red) in the absence of any friction as in \eqref{eq:markowitz} with a stochastic signal (purple). }
\label{fig:markowitz}
	\end{figure}

\begin{remark}
In Figure~\ref{fig:markowitz}, the smoothness of the optimal positions $X_t$ is mainly driven by market frictions, in particular the temporary impact $\Lambda$ in \eqref{exmp-lam}. Larger $\Lambda_{ii}$ increases the instantaneous cost of fast trading in Asset $i$ and therefore suppresses large trading speeds $u_t^i$, yielding smoother inventory paths. On the other hand, when $\Lambda_{ii}$ is small, the strategy can adjust much more rapidly towards the Markowitz portfolio. Furthermore,  the choice of the kernel $\phi$ matters most when $\Lambda_{ii}$ is small: Kernels with singularities at $0$ (e.g., fractional) impose stronger intertemporal dynamics of execution costs and further smooth the adjustment compared to non-singular kernels. By contrast, the risk aversion parameter $\gamma$ primarily controls the amplitude of the positions through the frictionless benchmark \eqref{eq:markowitz}, so that larger $\gamma$ reduces variability, but does not on its own penalize large trading speeds. 
\end{remark}


\subsection{Numerical scheme}\label{S:numericalscheme}

A major advantage of the operator formulas appearing in Theorem~\ref{thm:stochastic} is their ease of implementation. One possibility is the use of the so-called Nyström method to discretize the operators, and the method described in Section 5.1 of \cite{abijaber2022optimal} for the single-asset case can be readily adapted to our multi-asset setting. 

For completeness, we briefly describe the numerical implementation for deterministic signals as in Corollary~\ref{thm:deterministic}.

Fix  $N$ assets, $n\in \mathbb N$ and a partition $0=t_0<t_1<t_2<\ldots<t_n=T$ of $[0,T]$ as well as a propagator matrix in the form~\eqref{eq:Aphi}. We set $\Delta t \vcentcolon= T/n$.

\textbf{Step 1.} Specify the signal $P=\int_0^{\boldsymbol{\cdot}} I_s ds + M $ for a deterministic function $I$  and compute the $N(n+1)$-vector,
$$g_n:= (g(t_0)^\top, g(t_1)^\top,\ldots, g(t_n)^\top)^\top,$$
using the expression of $g$ in \eqref{eq:gtwithvarrho}. (Note that $g(t_i)\in \mathbb R^N$.) Second, specify the scalar convolution kernel $\phi:[0,T]\to \mathbb R$ that appears in \eqref{eq:Aphi}  and set $F_1(t,s)=\mathds{1}_{\{t> s\}}(T-t)$ and $F_2(t,s) = \mathds{1}_{\{t> s\}}$. Define the following lower and upper triangular $(n+1)\times (n+1)$-matrices $L_1,L_2,L$ and $U_1,U_2,U$ where the nonzero elements are given by:
\bq
L^{kj}_1 &=&  \int_{t_j}^{t_{j+1}}   F_1(t_k,s) ds =   (T-t_k) \Delta t , \quad k =0, \ldots, n, \quad j=0,\ldots, (k-1), \\
U^{kj}_1 &=&  \int_{t_j}^{t_{j+1}} F_1(s,t_k) ds { \approx } (T-t_j) \Delta t , \quad k=0, \ldots, n, \quad j=k, \ldots, (n-1), \\ 
L^{kj}_2 &=& \int_{t_j}^{t_{j+1}}   F_2(t_k,s) ds =  \Delta t, \quad k =0, \ldots, n, \quad j=0,\ldots, (k-1),\\
U^{kj}_2 &=& \int_{t_j}^{t_{j+1}} F_2(s,t_k) ds =   \Delta t , \quad k=0, \ldots, n, \quad j=k, \ldots, (n-1), \\
L^{kj} &=&\int_{t_j}^{t_{j+1}}   \phi(t_k-s) ds, \quad k =0, \ldots, n, \quad j=0,\ldots, (k-1),\\ 
U^{kj} &=&  \int_{t_j}^{t_{j+1}} \phi(s-t_k) ds, \quad k=0, \ldots, n, \quad j=k, \ldots, (n-1). 
\eq

For instance, for the three particular kernels in the previous subsection, $L$ and $U$ admit explicit expressions collected in Table \ref{T:kernels}.

\begin{table}[h!]
\centering
\resizebox{\textwidth}{!}{
\begin{tabular}{c c c c }
\hline\hline
&  & $L^{kj}$  & $U^{kj}$ \\ 
& $\phi(t)$ & for $0\leq j \leq k-1$ & for $k\leq j \leq n-1$ \\
\hline \hline \\[0.5ex]
No-transient		& 0 & 0 & 0\\ \\
Exponential	& $c{\rm e}^{-\rho t}$ & $  c \displaystyle\frac{e^{\rho \Delta t}- 1}{\rho} e^{-\rho(k-j)\Delta t}$ & {$ c\displaystyle\frac{1- e^{-\rho \Delta t}}{\rho}e^{-\rho(j-k)\Delta t}$}\\ \\
Fractional		& $c\,{t^{-\alpha}}$ & $ \frac{c (\Delta t)^{1-\alpha}}{1-\alpha}((k-j)^{1-\alpha} - (k-j-1)^{1-\alpha}  )$ & $\frac{c (\Delta t)^{1-\alpha}}{1-\alpha}((j+1-k)^{1-\alpha} - (j-k)^{1-\alpha}  )$\\ \\
\hline
\end{tabular}
}
\caption{Kernels $\phi$ and the corresponding explicit nonzero elements of the matrices $L$ and $U$.}
\label{T:kernels}
\end{table}

\textbf{Step 2.} Construct the $N(n+1)\times N(n+1)$-matrix $D_n$ using
$$ D_n = I_{(n+1)\times (n+1)} \otimes \bar \Lambda + L \otimes C + U \otimes C  + \gamma L_1 \otimes \Sigma + \gamma U_1 \otimes \Sigma  + \varrho L_2 \otimes \Pi + \varrho U_2\otimes \Pi, $$
where $\otimes$ is the Kronecker product. 

\textbf{Step 3.}   Recover the $N(n+1)$-vector for the optimal control path
$$ u_n = (D_n)^{-1} g_n $$
and note that the first $N$ components are the first trading values for all assets, the next $N$ components are the second trading values for all assets etc.
  
 \section{Systems of Stochastic Fredholm Equations} \label{sec-fred}
In this section we derive solutions to a class of coupled stochastic Fredholm equations of the second kind with both forward and backward components, which arise from the first-order condition in the proof of Theorem \ref{thm:stochastic} in Section \ref{sec-proof-strat}. The main result of this section is given in the following proposition. Recall that the class of admissible propagator matrices $\mathcal{G}$ was defined in \eqref{prop-class} and that the $\diamond$ operation was defined in Definition \ref{def:column-wise}. We also recall for any kernel $G\in \mathcal G$, the kernel $G_t$ is the truncation of $G$ from Definition \ref{trunc-ker} and that $\mathbf{G}_t$ is the operator induced by the kernel $G_t$.  

\begin{prop}\label{prop:Fredholm}
Let $K,L$ be Volterra kernels in $\mathcal{G}$, let $\bar{\Lambda}\in\R^{N\times N}$ be a symmetric positive definite matrix, and let $f=(f_t)_{0\leq t\leq T}$ be an $N$-dimensional progressively measurable process with $\smash{\int_0^T\E[\|f_t\|^2]dt<\infty}$. Then the following coupled system of stochastic Fredholm equations
\be \label{stoch-fred}
\bar{\Lambda}u_t=f_t-\int_0^t K(t,r)u_rdr-\int_t^T L(r,t)^\top\E_t[u_r]dr,\quad t\in[0,T],
\ee
admits a unique progressively measurable solution $u^*=(u^*_t)_{0\leq t\leq T}$ in $\mathcal{U}$ given by 
\be
u_t^*=\big((\mathbf{I} + \mathbf{B})^{-1}a\big)(t),\quad t\in[0,T],
\ee
where
\be
a_t=\bar{\Lambda}^{-1}\Big(f_t-\int_t^TL(r,t)^\top \big(\mathbf{D}_t^{-1}\mathds{1}_{\{t\leq\cdot\}}\E_t[f_\cdot]\big)(r)dr\Big),
\ee
\be
B(t,s)=-\bar{\Lambda}^{-1}\Big(\mathds{1}_{\{s\leq t\}}\int_t^T L(r,t)^\top \big(\mathbf{D}_t^{-1}\diamond\mathds{1}_{\{t\leq \cdot\}}K(\cdot,s)\big)(r)dr-K(t,s)\Big),
\ee
\be
\mathbf{D}_t=\bar{\Lambda}\mathbf{I}+\mathbf{K}_t+\mathbf{L}_t^*. 
\ee
Here $\mathbf{I}$ is the identity operator on $L^2([0,T],\R^N)$ and $\mathbf{B}$ is the integral operator induced by the kernel $B$.
\end{prop}
\begin{remark}\label{rem:composition of positive operators}
Proposition \ref{prop:Fredholm} is a generalization of Proposition 5.1 in \cite{abijaber2023equilibrium} from the class of one-dimensional stochastic Fredholm equations to $N$--dimensional coupled systems. Note that in the one-dimensional case $\bar{\Lambda}$ is just a positive scalar, so that the integral operators $\bar{\Lambda}^{-1}\mathbf{K}_t$ and $\bar{\Lambda}^{-1}\mathbf{L}_t^*$ are nonnegative definite and therefore the operator $\mathbf{I}+\bar{\Lambda}^{-1}(\mathbf{K}_t+\mathbf{L}_t^*)$ invertible, which is an important ingredient in the derivation of the solution. However, since the composition of two general nonnegative definite linear operators on a Hilbert space is not necessarily nonnegative definite, even if they are self-adjoint, this argument does not go through in the $N$--dimensional case. Thus, the positive definite operator $\mathbf{D}_t=\bar{\Lambda}\mathbf{I}+\mathbf{K}_t+\mathbf{L}_t^*$ is inverted as part of the solution, and $\bar{\Lambda}^{-1}$ appears as part of the terms $a$ and $B$.
\end{remark}
\noindent As a preparation for the proof, we introduce the concept of  kernels of type $L^2$  following the terminology of Definition 9.2.2 in \cite{gripenberg1990volterra}.  
\begin{definition}\label{def:type L2}
A Borel-measurable function $G:[0,T]^2\to\R^{N\times N}$ is called a kernel of type $L^2$ if it satisfies
\be
\big\|G\big\|_{L^2}\vcentcolon= \sup\limits_{\substack{f:\|f\|_{L^2}\leq 1\\g:\|g\|_{L^2}\leq 1}}\int_0^T\int_0^T\big\|g(t)G(t,s)f(s)\big\|dsdt<\infty,
\ee
where the supremum is taken over all real-valued functions $f,g\in L^2([0,T],\R)$ with norm bounded by $1$ and $\|\cdot\|$ denotes the Frobenius norm (following the convention introduced after \eqref{l-2-ker-def}). 
\end{definition}
Every kernel in $L^2([0,T]^2,\R^{N\times N})$ is also a kernel of type $L^2$ (see Proposition 9.2.7 (iii) in \cite{gripenberg1990volterra}). However, the converse does not hold in general. 
The following auxiliary lemmas are essential ingredients for the proof of Proposition \ref{prop:Fredholm}.
\begin{lemma}\label{lemma:Kt+Lt* pos semidefinite}
For $K,L\in\mathcal{G}$ and any fixed $t\in[0,T]$, the integral operator $\mathbf{K}_t+\mathbf{L}_t^*$ is nonnegative definite.
\end{lemma}
\begin{proof}
Let $f\in L^2([0,T],\R^N)$. From Definition \ref{trunc-ker} we have, 
\bn
&&\langle f,(\mathbf{K}_t+\mathbf{L}_t^*)f\rangle_{L^2}\\
&&=\int_0^T\int_0^T f(s)^\top\Big(\mathds{1}_{\{r\geq t\}}K(s,r)+\mathds{1}_{\{s\geq t\}}L(r,s)^\top\Big)f(r)drds \\
&&=\int_0^T\int_0^T f(s)^\top\Big(\mathds{1}_{\{r\geq t\}}\mathds{1}_{\{s\geq r\}}\mathds{1}_{\{s\geq t\}}K(s,r)+\mathds{1}_{\{s\geq t\}}\mathds{1}_{\{r\geq s\}}\mathds{1}_{\{r\geq t\}}L(r,s)^\top\Big)f(r)drds \\
&&=\int_0^T\int_0^T \mathds{1}_{\{s\geq t\}}f(s)^\top\Big(K(s,r)+L(r,s)^\top\Big)\mathds{1}_{\{r\geq t\}}f(r)drds\\
&&=\langle f_t,(\mathbf{K}+\mathbf{L}^*)f_t\rangle_{L^2}\geq 0,
\en
where $f_t(s)\vcentcolon =\mathds{1}_{\{s\geq t\}}f(s)$ and we have used the fact that $K$ and $L$ are Volterra kernels in $\mathcal{G}$ (see \eqref{prop-class}).
\end{proof}

\begin{lemma}\label{lemma:lambda positive} 
Let $\bar{\Lambda}\in\R^{N\times N}$ be a symmetric positive definite matrix. Then, 
the linear operator $\bar{\Lambda}\mathbf{I}$ on $L^2([0,T],\R^N)$ is self-adjoint and satisfies 
\be
\langle \bar{\Lambda}\mathbf{I}f,f\rangle_{L^2}=
\langle \bar{\Lambda}f,f\rangle_{L^2}\geq \lambda_{\text{min}}\langle f,f\rangle_{L^2}, \quad \textrm{for all } f\in L^2([0,T],\R^N)
\ee
where $\lambda_{\text{min}}>0$ denotes the smallest eigenvalue of $\bar{\Lambda}$. In particular, $\bar{\Lambda}\mathbf{I}$ is positive definite.
\end{lemma}
\begin{proof}
Since $\bar{\Lambda}$ is symmetric positive definite, it can be decomposed as
\be
\bar{\Lambda}=Q\Delta Q^\top,
\ee
where $Q$ is an orthogonal matrix whose columns are orthonormal eigenvectors of $\bar{\Lambda}$ and $\Delta$ 
is a diagonal matrix whose entries $(\lambda_i)_{i=1,...,N}$ are the positive eigenvalues of $\bar{\Lambda}$. Let $x\in\R^N$ and define $y=Q^\top x$, which in particular implies that $\|y\|=\|x\|$ as $Q$ is orthogonal. We get, 
\be \label{s-eq} 
x^\top\bar{\Lambda}x=x^\top Q\Delta Q^\top x 
=\|\Delta^{\frac{1}{2}}Q^\top x\|^2
=\sum_{i=1}^N \lambda_iy_i^2
\geq \sum_{i=1}^N \lambda_{\text{min}}y_i^2
=\lambda_{\text{min}}\| x\|^2.
\ee
Let $f\in L^2([0,T],\R^N)$. Then from \eqref{s-eq} it follows that
\be
\langle f, \bar{\Lambda}f\rangle_{L^2}
=\int_0^Tf(t)^\top\bar{\Lambda}f(t)
\geq \int_0^T \lambda_{\text{min}}\|f(t)\|^2 dt
=\lambda_{\text{min}}\langle f,f\rangle_{L^2}.
\ee
\end{proof}
 
\textbf{Notation.} Let $E:V\to W$ be a linear operator between two normed real vector spaces $(V,\|\cdot\|_V)$ and $(W,\|\cdot\|_W)$. Then the operator norm of $E$ is denoted as,
\be \label{eq:operator norm}
\|E\|_{\text{op}}\vcentcolon = \sup\big\{\|E(x)\|_W : x\in V\text{ with }\|x\|_V\leq 1\big\}.
\ee
\begin{lemma} \label{lemma:inverse operator}
Let $E$ be a bounded linear operator from a real Hilbert space $V$
into itself. Suppose that there exists a constant $c>0$ such that $\langle Ex,x\rangle \geq c \langle x,x\rangle$ for all $x\in  V$. Then $E$ is invertible and $\|E^{-1}\|_{\text{op}}\leq c^{-1}$.
\end{lemma} 
\begin{proof}
    See \cite{torchinsky2015problems}, Chapter 10, Problem 165.
\end{proof}

Now we are ready to prove Proposition \ref{prop:Fredholm}.
\begin{proof}[Proof of Proposition \ref{prop:Fredholm}] For every $0\leq t \leq T$ define the auxiliary process
\be
m_t(s) =\mathds{1}_{\{t\leq s\}}\E_t[u_s],\quad 0\leq s \leq T. 
\ee
Taking the conditional expectation $\E_t[\cdot]$ on both sides of \eqref{stoch-fred}, then multiplying by $\mathds{1}_{\{t\leq s\}}$ and using the tower property we get
\be
\begin{aligned}\label{m_t-equation}
\bar{\Lambda}m_t(s)&=\mathds{1}_{\{t\leq s\}}\E_t[f_s]-\mathds{1}_{\{t\leq s\}}\int_0^t K(s,r)u_rdr \\
&\quad -\mathds{1}_{\{t\leq s\}}\int_t^s K(s,r)\E_t[u_r]dr-\mathds{1}_{\{t\leq s\}}\int_s^T L(r,s)^\top\E_t[u_r]dr \\ 
&=f_t^u(s)-\int_t^s K_t(s,r)m_t(r)dr-\int_s^T L_t(r,s)^\top m_t(r)dr \\ 
&=f_t^u(s)-\big((\mathbf{K}_t+\mathbf{L}_t^*)m_t\big)(s), \quad \textrm{for all } 0\leq s\leq T, 
\end{aligned}
\ee
where
\be \label{f-term}
f_t^u(s)\vcentcolon = \mathds{1}_{\{t\leq s\}}\E_t[f_s]-\mathds{1}_{\{t\leq s\}}\int_0^t K(s,r)u_rdr.
\ee
The bounded linear operator $\mathbf{K}_t+\mathbf{L}_t^*$ is nonnegative definite by Lemma \ref{lemma:Kt+Lt* pos semidefinite}. Thus, due to  Lemma \ref{lemma:lambda positive} and Lemma \ref{lemma:inverse operator}, the operator $\mathbf{D}_t=\bar{\Lambda}\mathbf{I}+\mathbf{K}_t+\mathbf{L}_t^*$ is invertible with, 
\be \label{d-inv-bnd} 
\|\mathbf{D}_t^{-1}\|_{\text{op}}\leq\lambda_{\text{min}}^{-1}, \quad \textrm{for all } 0\leq t\leq T, 
\ee
where $\|\cdot\|_{\text{op}}$ denotes the operator norm as introduced in \eqref{eq:operator norm} and $\lambda_{\text{min}}>0$ denotes the smallest eigenvalue of $\bar{\Lambda}$. It follows from \eqref{m_t-equation} that  
\be \label{m-expr}
m_t(s)=(\mathbf{D}_t^{-1}f_t^u)(s), \quad \textrm{for all } 0 \leq s \leq T.
\ee
By plugging in \eqref{m-expr} into \eqref{stoch-fred}  we get 
\be \label{r-1} 
\bar{\Lambda}u_t=f_t-\int_0^t K(t,r)u_rdr-\int_t^T L(r,t)^\top(\mathbf{D}_t^{-1}f_t^u)(r)dr.
\ee
Using \eqref{f-term} we rewrite the third term on the right-hand side of \eqref{r-1} as follows: 
\be \label{gg-1}
\begin{aligned} 
&\int_t^T L(r,t)^\top(\mathbf{D}_t^{-1}f_t^u)(r)dr \\  
&=\int_t^TL(r,t)^\top \big(\mathbf{D}_t^{-1}\mathds{1}_{\{t\leq\cdot\}}\E_t[f_\cdot]\big)(r)dr 
-\int_t^TL(r,t)^\top \big(\mathbf{D}_t^{-1}\mathds{1}_{\{t\leq\cdot\}}\int_0^t K(\cdot,s)u_sds\big)(r)dr \\ 
&=\int_t^TL(r,t)^\top \big(\mathbf{D}_t^{-1}\mathds{1}_{\{t\leq\cdot\}}\E_t[f_\cdot]\big)(r)dr \\
&\quad -\int_0^t\int_t^T L(r,t)^\top \big(\mathbf{D}_t^{-1}\diamond\mathds{1}_{\{t\leq \cdot\}}K(\cdot,s)\big)(r)dru_sds,
\end{aligned} 
\ee
where $\mathbf{D}_t^{-1}\diamond\mathds{1}_{\{t\leq \cdot\}}K(\cdot,s)$ denotes the column-wise application of the operator $\mathbf{D}_t^{-1}$ to the function $\mathds{1}_{\{t\leq r\}}K(r,s)$ for fixed $s\in[0,T]$ (see Definition \ref{def:column-wise}). The above equality holds because of Fubini's theorem and the fact that the operator $\mathbf{D}_t^{-1}$, multiplication by $\mathds{1}_{\{t\leq r\}}K(r,s)$, and integration are all linear operations. From \eqref{r-1} and \eqref{gg-1} it follows that  
\be \label{gg-2} 
u_t=a_t-\int_0^T B(t,s)u_sds, \quad 0\leq t\leq T,
\ee
where $a(\cdot)$ and $B(\cdot,\cdot)$ are given in the statement of Proposition \ref{prop:Fredholm}. 

We introduce the following auxiliary lemma, which will be proved at the end of this section.   
\begin{lemma}\label{lemma:integrability}
Let $a,B$ be defined as in Proposition \ref{prop:Fredholm} and let $u$ be a solution of \eqref{gg-2}. Then the following hold: 
\begin{enumerate}
    \item [\textbf{(i)}] $\E\big[\int_0^T\|a_t\|^2dt\big]<\infty$.
    \item  [\textbf{(ii)}] $\int_0^T\int_0^T \|B(t,s)\|^2dsdt<\infty$.
    \item  [\textbf{(iii)}] $\E\big[\int_0^T\|u_t\|^2dt\big]<\infty$.
\end{enumerate}
\end{lemma}
Note that by Corollary 9.3.16 of \cite{gripenberg1990volterra} and Lemma \ref{lemma:integrability}(ii), the deterministic Volterra kernel $B$ has a resolvent $R^B:[0,T]^2\to\R^{N\times N}$ of type $L^2$ (recall Definition \ref{def:type L2}), which is the kernel corresponding to this operator,  
\be
\mathbf{R}^B=\mathbf{I}-(\mathbf{I}+\mathbf{B})^{-1}.
\ee
Therefore, by Theorem 9.3.6 of \cite{gripenberg1990volterra} the integral equation \eqref{gg-2} admits a unique solution $u^*(\omega)$ for any fixed $\omega \in \Omega$ for which
$ a(\omega)\in L^2([0,T],\R^N)$, i.e. for $\P$-a.e. $\omega$ by Lemma \ref{lemma:integrability}(i). This solution $u^*$ is given by a variation of constants formula in terms of the resolvent $R^{B}$, namely
\be \label{gf1} 
u_t^*=a_t -\int_0^T R^B(t,s)a_sds,\quad \quad 0\leq t\leq T.
\ee

It follows that any solution $u^*\in\mathcal{U}$ to the coupled system of stochastic Fredholm equations must be of the form
\be
u_t^*=\big((\mathbf{I} +\mathbf{B})^{-1}a\big)(t), \quad 0\leq t \leq T, \ \P-\rm{a.s.}, \label{eq:solution}
\ee
which yields the uniqueness of $u^*$ up to modifications. For the existence, let $u^*$ be defined as in \eqref{eq:solution}. It follows that $u^*$ satisfies \eqref{gf1} and thus \eqref{gg-2}. Therefore, it also satisfies \eqref{r-1} and consequently \eqref{stoch-fred} due to \eqref{m-expr}, as desired. Thanks to Lemma \ref{lemma:integrability}(iii) it is ensured that $u^*\in\mathcal{U}$. 
\end{proof}

 \begin{proof}[Proof of Lemma \ref{lemma:integrability}] (i) Recall that 
\be
a_t=\bar{\Lambda}^{-1}\Big(f_t-\int_t^TL(r,t)^\top \big(\mathbf{D}_t^{-1}\mathds{1}_{\{t\leq\cdot\}}\E_t[f_\cdot]\big)(r)dr\Big).
\ee
Since $ C_1:= {\E\big[\int_0^T\|f_t\|^2dt\big]<\infty}$ by the hypothesis of the lemma, it suffices to show that 
\be \label{gg-66} 
\E\left[ \int_0^T \left\|\int_t^TL(r,t)^\top \big(\mathbf{D}_t^{-1}\mathds{1}_{\{t\leq\cdot\}}\E_t[f_\cdot]\big)(r)dr\right\|^2dt\right]<\infty.
\ee
To see this, note  that from \eqref{d-inv-bnd}, from the conditional Jensen inequality and the tower property we get
 \be \label{gg6} 
 \begin{aligned} 
\E\Big[ \int_0^T\big\|\big(\mathbf{D}_t^{-1}\mathds{1}_{\{t\leq\cdot\}}\E_t[f_\cdot]\big)(r)\big\|^2 dr\Big] 
&\leq\lambda_{\text{min}}^{-2}\E\Big[ \int_0^T\big\|\mathds{1}_{\{t\leq r\}}\E_t[f_r]\big\|^2 dr\Big] \\ 
&\leq \lambda_{\text{min}}^{-2}\int_0^T\E\Big[ \E_t\big[\|f_r\|^2\big]\Big] dr\\
& \leq \lambda_{\text{min}}^{-2} C_1,  \quad \textrm{for all } 0\leq t\leq T. 
\end{aligned} 
\ee
Moreover, recall that $L\in\mathcal{G}$ by assumption and thus 
$$C_2 :=\int_0^T\int_0^T \|L(r,t)\|^2drdt <\infty,$$ see \eqref{prop-class}.
The submultiplicativity of the Frobenius norm,  Hölder's inequality and \eqref{gg6}  yield
\be
 \begin{aligned} 
&\E\left[ \int_0^T \left\|\int_t^TL(r,t)^\top \big(\mathbf{D}_t^{-1}\mathds{1}_{\{t\leq\cdot\}}\E_t[f_\cdot]\big)(r)dr\right\|^2dt\right] \\
&\leq \E\left[ \int_0^T\Big(\int_0^T \big\|L(r,t)\big\| \big\|\big(\mathbf{D}_t^{-1}\mathds{1}_{\{t\leq\cdot\}}\E_t[f_\cdot]\big)(r)\big\|dr\Big)^2
dt\right] \\
&=  \int_0^T\int_0^T \big\|L(r,t)\big\|^2dr \hspace{1mm}\E\left[\int_0^T \big\|\big(\mathbf{D}_t^{-1}\mathds{1}_{\{t\leq\cdot\}}\E_t[f_\cdot]\big)(r)\big\|^2dr\right]
dt \\ 
&\leq C_2C_1, 
\end{aligned} 
\ee
which verifies \eqref{gg-66}.

(ii) Recall that 
\be \label{b-def} 
B(t,s)= -\bar{\Lambda}^{-1}\Big(\mathds{1}_{\{s\leq t\}}\int_t^T L(r,t)^\top \big(\mathbf{D}_t^{-1}\diamond\mathds{1}_{\{t\leq \cdot\}}K(\cdot,s)\big)(r)dr-K(t,s)\Big).
\ee
Since $K\in\mathcal{G}$, it holds that $C_3 := \int_0^T\int_0^T \|K(r,s)\|^2dsdr <\infty$. Moreover, the submultiplicativity of the Frobenius norm, Hölder's inequality and \eqref{d-inv-bnd} imply that
\be
\begin{aligned} 
&\int_0^T\int_0^T \left\|\mathds{1}_{\{s\leq t\}}\int_t^T L(r,t)^\top \big(\mathbf{D}_t^{-1}\diamond\mathds{1}_{\{t\leq \cdot\}}K(\cdot,s)\big)(r)dr\right\|^2 dsdt \\ 
&\leq\int_0^T\int_0^T\Big(\int_0^T \big\|L(r,t)\big\| \big\|\big(\mathbf{D}_t^{-1}\diamond\mathds{1}_{\{t\leq \cdot\}}K(\cdot,s)\big)(r)\big\|dr\Big)^2 dsdt \\
&\leq\int_0^T\int_0^T\int_0^T \big\|L(r,t)\big\|^2dr \int_0^T\big\|\big(\mathbf{D}_t^{-1}\diamond\mathds{1}_{\{t\leq \cdot\}}K(\cdot,s)\big)(r)\big\|^2dr ds dt \\ 
&\leq\int_0^T\int_0^T\int_0^T \big\|L(r,t)\big\|^2dr \lambda_{\text{min}}^{-2}\int_0^T\big\|K(r,s)\big\|^2dr ds dt \\
&\leq \lambda_{\text{min}}^{-2}C_2 C_3.
\end{aligned} 
\ee
Together with \eqref{b-def} we get the result. 

(iii)  From (i) it holds that $ C_4:=\E[\int_0^T\|a_t\|^2dt]<\infty$.  Furthermore, since $R^B$ is of type $L^2$, we have that $C_5:=\|R^B\|^2_{L^2} <\infty$. Now it follows from an $L^p$-inequality for kernels (see \cite{gripenberg1990volterra}, Theorem 9.2.4) that
\bq
\E\left[\int_0^T\left\|\int_0^T R^B(t,s)a_sds\right\|^2dt\right] &=&\E\left [\left\|\int_0^T R^B(t,s)a_sds \right\|_{L^2}^2 \right] \\
&\leq & \E\Big[\big\|R^B\big\|^2_{L^2}\big\|a\big\|^2_{L^2}\Big] \\
&=&C_5C_4. 
\eq
Together with \eqref{admis} and \eqref{gf1} it follows that $u^*\in\mathcal{U}$. 
\end{proof}

\section{Proofs of Theorem \ref{thm:stochastic} and Corollary \ref{thm:deterministic}} \label{sec-proof-strat}

In order to prove Theorem \ref{thm:stochastic} and Corollary \ref{thm:deterministic}, we introduce the following essential lemmas. 
We first prove that the objective functional $J(u)$ in \eqref{eq:Ju} is strictly concave in $u\in\mathcal{U}$. 
\begin{lemma}\label{lemma:concavity}
For any fixed $G\in\mathcal{G}$ the map $u\mapsto J(u)$ is strictly concave in $u\in\mathcal{U}$. 
\end{lemma}
\begin{proof}
We prove the result by showing that each of the components of $J(u)$ is concave and at least one of them is strictly concave. We first notice that the term
\be
-\int_0^T u_t^\top P_tdt+ (X_T^u)^\top P_T,
\ee
is affine in $u$.

Recall that $\Lambda$ is a positive definite matrix, hence the quadratic form $x\mapsto x^\top\Lambda x$ is strictly convex in $x\in\R^N$, which implies that
\be
-\int_0^T u^\top_t\Lambda u_tdt
\ee
is strictly concave in $u$ due to the linearity and monotonicity of the integral. Thus by \eqref{eq:Ju}, $J(u)$ is strictly concave if the term
\be \label{c-12} 
-\int_0^Tu_t^\top D_t^udt-\frac{\gamma}{2}\int_0^T(X_t^u)^\top\Sigma X_t^udt =:C_1(u) + C_2(u), 
\ee
is concave in $u$. Without loss of generality, we assume that $X_0=0$, as terms in $C_2(\cdot)$ involving $X_0$ are affine in $u$.  
In order to show the concavity of \eqref{c-12}, we adopt the approach from Gatheral et al. (see  Proposition 2.9 in \cite{GSS}).  
Since $G\in\mathcal{G}$ and $\Sigma$ is nonnegative definite, it follows from \eqref{dist} and \eqref{non-neg} that $C_1(u),C_2(u)\leq 0$ for all $u\in\mathcal{U}$. Next, define the cross functionals
\be \label{c23} 
\wt C_1(u,v):= -\int_0^T  u_t^\top D_t^vdt\quad\text{and}\quad
\wt C_2(u,v):= -\frac{\gamma}{2}\int_0^T (X_t^u)^\top \Sigma X^v_tdt, \quad u,v\in\mathcal{U}. 
\ee
From \eqref{c-12} and \eqref{c23} it follows that 
\be
C_i(u-v)=C_i(u)+C_i(v)- \wt C_i(u,v)- \wt C_i(v,u), \quad i=1,2
\ee
and since $C_i(u-v)\leq 0$, $ i=1,2$ we get that
\be \label{c34} 
C_i\left(\frac{1}{2}u+\frac{1}{2}v\right)
=\frac{1}{4}C_i(u)+\frac{1}{4}C_i(v)+\frac{1}{4}\wt C_i(u,v)+\frac{1}{4}\wt C_i(v,u)
\geq \frac{1}{2}C_i(u)+\frac{1}{2}C_i(v).
\ee
The concavity of $C_i$ for $i=1,2$, follows from \eqref{c34} and the fact that the map
\be
\lambda \mapsto C_i(\lambda u+(1-\lambda)v)=\lambda^2C_i(u)+(1-\lambda)^2C_i(v)+\lambda(1-\lambda)(\wt C_i(u,v)+\wt C_i(v,u))
\ee
is continuous in $\lambda$. We therefore obtain the concavity of the left-hand side of \eqref{c-12} and hence of $J(\cdot)$. 
\end{proof}
From Lemma \ref{lemma:concavity} it follows that the objective functional $J$ in \eqref{eq:Ju} admits a unique maximizer characterized by the critical point $u\in\mathcal{U}$ at which its G\^ateaux derivative
\be
\langle J'(u),v\rangle\vcentcolon =\lim\limits_{\varepsilon\to 0}\frac{J(u+\varepsilon v)-J(u)}{\varepsilon}
\ee
vanishes for any $v \in\mathcal{U}$ (see Propositions 1.2 and 2.1 in Chapter 2 of \cite{ekeland1999convex}). In the following lemma we calculate $\langle J'(u),v\rangle$. 
\begin{lemma}\label{lemma:gateaux}
For any $u,v\in\mathcal{U}$, the G\^ateaux derivative $\langle J'(u),v\rangle$ of $J$ is given by
\be\label{eq:gateaux}
\E\left[ \int_0^Tv_t^\top\Big(-P_t-\big((\mathbf{G}+\mathbf{G}^*)u\big)(t)-\bar{\Lambda}u_t+P_T-\gamma\int_t^T \Sigma X_s^uds  \Big)dt\right].
\ee
\end{lemma}
\begin{proof}
Let $u,v\in\mathcal{U}$ and let $\varepsilon>0$. Recall that the objective functional is given by
\be
J(u) = \mathbb{E}\bigg{[}\int_0^T -u_t^\top(P_t + D_t^u +\frac{1}{2}\Lambda u_t)dt+(X_T^u)^\top P_T-\frac{\gamma}{2}\int_0^T(X_t^u)^\top\Sigma X_t^udt\bigg{]}.
\ee
A direct computation using \eqref{inv} and \eqref{dist} yields
\bq
&&J(u+\varepsilon v)-J(u) \\
&&=\varepsilon\cdot\E \bigg[\int_0^T\Big(-v_t^\top(P_t + D_t^u +\frac{1}{2}\Lambda u_t)-u_t^\top( D_t^v +\frac{1}{2}\Lambda v_t) \Big)dt \\ 
&& \qquad\qquad+\int_0^T v_t^\top P_Tdt-\frac{\gamma}{2}\int_0^T \Big(\int_0^t v_s^\top ds \Big)\Sigma X_t^u -(X_t^u)^\top \Sigma \int_0^t v_s ds dt\bigg]+O(\varepsilon^2) \\ 
&&=\varepsilon\cdot \E\bigg{[} \int_0^Tv_t^\top\bigg(-P_t-\big((\mathbf{G}+\mathbf{G}^*)u\big)(t)-\frac{1}{2}\big(\Lambda+\Lambda^\top\big)u_t+P_T \\
&&\qquad\qquad -\gamma\int_t^T \Sigma X_s^uds  \bigg)dt\bigg{]} +O(\varepsilon^2),
\eq
where we used the symmetry of $\Sigma$ and applied Fubini's theorem in the last equality. Recalling \eqref{lam-bar}, dividing by $\varepsilon$ and taking the limit as $\varepsilon \dr 0$ yields the result.
\end{proof}
Next, we derive from Lemma \ref{lemma:gateaux} a system of stochastic Fredholm equations which is satisfied by the unique maximizer of the objective functional $J(u)$ in \eqref{eq:Ju}.
\begin{lemma}\label{lemma:fred-foc}
For fixed $G\in\mathcal{G}$, the unique maximizer of the objective functional $J(u)$ in \eqref{eq:Ju} satisfies the coupled system of stochastic Fredholm equations of the second kind given by
\be \label{fred-foc}
\bar{\Lambda}u_t=g_t-\int_0^t K(t,s)u_sds-\int_t^T K(s,t)^\top\E_t[u_s]ds,\quad 0\leq t\leq T, 
\ee
where $\bar{\Lambda}$, $K$ and $g$ are defined in \eqref{lam-bar}, \eqref{h-k-ker} and \eqref{g-def}, respectively. 
\end{lemma}
\begin{proof}
The unique maximizer of $J(u)$ is characterized by the critical point $u\in\mathcal{U}$ at which the G\^ateaux derivative $\langle J'(u),v\rangle$ is equal to 0 for any $v \in\mathcal{U}$. It follows from Lemma \ref{lemma:gateaux} that this is equivalent to
\be\label{eq:gateaux=0}
\E\left[ \int_0^Tv_t^\top\Big(-P_t-\big((\mathbf{G}+\mathbf{G}^*)u\big)(t)-\bar{\Lambda}u_t+P_T-\gamma\int_t^T \Sigma X_s^uds  \Big)dt\right]=0,\quad\forall v\in\mathcal{U}.
\ee
Next, by an application of Fubini's theorem and recalling that $F$ was defined in \eqref{h-k-ker},
\be
\begin{aligned}
\label{eq:inventory}
&\gamma\int_t^T \Sigma X_s^uds=\gamma\Sigma X_0(T-t)+ \gamma\Sigma \int_t^T\int_0^su_rdrds\\
&=\gamma\Sigma X_0(T-t)+\gamma\Sigma\int_0^T\big(T-(r\vee t)\big)u_rdr\\
&=\gamma\Sigma X_0(T-t)+\big((\mathbf{F}+\mathbf{F}^*)u\big)(t),
\end{aligned}
\ee
where $r\vee t\vcentcolon =\max\{r,t\}$. Plugging \eqref{eq:inventory} into \eqref{eq:gateaux=0}, conditioning on $\mathcal{F}_t$ and using the tower property of conditional expectation we get the following first-order condition,
\be
\begin{aligned}\label{eq:foc}
&-P_t-\big((\mathbf{G}+\mathbf{G}^*)(\E_tu)\big)(t)-\bar{\Lambda}u_t+\E_t[P_T]-\gamma\Sigma X_0(T-t)-\big((\mathbf{F}+\mathbf{F}^*)(\E_tu)\big)(t) \\
&=0, \quad d\mathbb{P}\otimes dt \text{-a.e. on }\Omega\times [0,T].
\end{aligned}
\ee
Recalling the definition of $K$ in \eqref{h-k-ker} and $g$ in \eqref{g-def}, equation \eqref{eq:foc} simplifies to
\be
\bar{\Lambda}u_t=g_t-\big((\mathbf{K}+\mathbf{K^*})(\E_tu)\big)(t), \quad d\mathbb{P}\otimes dt \text{-a.e. on }\Omega\times [0,T],
\ee
which is equivalent to \eqref{fred-foc}, since $K$ is a Volterra kernel.
\end{proof}
In order to prove Theorem \ref{thm:stochastic} and Corollary \ref{thm:deterministic}, the following lemma is needed, which in combination with Remark \ref{rem:pos semidefinite} shows that the Volterra kernel $F$ defined in $\eqref{h-k-ker}$ is in $\mathcal{G}$.
\begin{lemma}\label{lemma:F+F^* nonnegative def}
The following inequality holds, 
\be
\int_0^T\int_0^Tf(t)^\top\gamma\Sigma \big(T-(s\vee t)\big)f(s)dsdt\geq 0, \quad \textrm{for all } f\in L^2([0,T],\R^N). 
\ee
\end{lemma}
\begin{proof}
First, observe that for all $(s,t)\in [0,T]^2$ we have
\be
T-(s\vee t)=(T-s)\wedge (T-t)
=\int_0^T \mathds{1}_{\{r\leq T-s\}}\mathds{1}_{\{r\leq T-t\}}dr
=\int_0^T \mathds{1}_{\{s\leq T-r\}}\mathds{1}_{\{t\leq T-r\}}dr.
\ee
Thus for any $f\in L^2([0,T],\R^N)$ we get that
\bq
&&\int_0^T\int_0^Tf(t)^\top\gamma\Sigma \big(T-(s\vee t)\big)f(s)dsdt \\
&&=\gamma\int_0^T\int_0^{T-r}\hspace{-3mm}\int_0^{T-r}\hspace{-3mm}f(t)^\top\Sigma f(s)dtdsdr \\
&&=\gamma \int_0^T \Big(\int_0^{T-r}\hspace{-3mm}f(t)dt\Big)^\top\Sigma \Big(\int_0^{T-r}\hspace{-3mm}f(t)dt\Big)dr \\
&&\geq 0,
\eq
since $\Sigma$ is nonnegative definite and $\gamma\geq0$.
\end{proof}

Now we are ready to prove Theorem \ref{thm:stochastic} and Corollary \ref{thm:deterministic}. 

\begin{proof}[Proof of Theorem \ref{thm:stochastic}] 
Let $G\in\mathcal{G}$. Recall that $F$, $K$ and $g$ are defined in \eqref{lam-bar}, \eqref{h-k-ker} and \eqref{g-def}, respectively. It follows from Lemma \ref{lemma:F+F^* nonnegative def} and Remark \ref{rem:pos semidefinite} that $F\in\mathcal{G}$ and thus $K\in\mathcal{G}$ as well. Moreover, $g$ is progressively measurable and satisfies
\be\int_0^T\E\big[\|g_t\|^2\big]dt<\infty,\ee 
because $P$ is progressively measurable and satisfies $\smash{\int_0^T\E[\|P_t\|^2]dt<\infty}$ by assumption.

Now by Lemmas \ref{lemma:concavity} and \ref{lemma:fred-foc}, the objective functional $J(u)$ in \eqref{eq:Ju} admits a unique maximizer in $\mathcal{U}$, and this maximizer satisfies \eqref{fred-foc}. An application of Proposition \ref{prop:Fredholm} with $f=g$ and $L=K$ yields that \eqref{fred-foc} admits a unique progressively measurable solution $u^*\in\mathcal{U}$ given by 
\be
u_t^*=\big((\mathbf{I}+\mathbf{B})^{-1}a\big)(t),\quad 0\leq t\leq T, \label{eq:solution u^*}
\ee
with $a$ and $\mathbf{B}$ defined as in Theorem \ref{thm:stochastic}. It follows that $u^*$ in \eqref{eq:solution u^*} must be the unique maximizer of $J(u)$, which finishes the proof.
\end{proof}
\begin{proof}[Proof of Corollary \ref{thm:deterministic}]
Let $G \in \mathcal G$ and assume that $A$ in \eqref{p-dec} is deterministic. Then, $g$ in \eqref{g-def} is a deterministic function in $L^2([0,T],\R^N)$ given by
\be
g_t=A_T-A_t-\gamma (T-t)\Sigma X_0,\quad 0\leq t\leq T.
\ee
It follows that the unique maximizer $u^*$ of the objective functional $J(u)$ from Theorem \ref{thm:stochastic}, which satisfies \eqref{fred-foc}, is deterministic and satisfies the simplified equation
\be\label{eq:deterministic equation}
\bar{\Lambda}u_t = g_t-\big((\mathbf{K}+\mathbf{K}^*)u\big)(t),\quad 0\leq t\leq T.
\ee
Now \eqref{eq:deterministic equation} admits the deterministic solution 
\be\label{eq:deterministic solution}
u_t^*=\big(\mathbf{D}^{-1}g\big)(t),\quad 0\leq t \leq T,
\ee
where the bounded linear operator 
\be
\mathbf{D}=(\mathbf{K}+\mathbf{K}^*+\bar{\Lambda}\mathbf{I})=(\mathbf{G}+\mathbf{G}^*+\mathbf{F}+\mathbf{F}^*+\bar{\Lambda}\mathbf{I})
\ee
defined in \eqref{h-k-ker} and \eqref{D-def}
is invertible because of the following: Since $G\in\mathcal{G}$ by assumption, $F\in \mathcal{G}$ due to Lemma \ref{lemma:F+F^* nonnegative def} and Remark \ref{rem:pos semidefinite}, and there exists a real number $c>0$ such that $\bar{\Lambda}\mathbf{I}$ satisfies 
\be
\langle \bar{\Lambda}\mathbf{I}f,f\rangle_{L^2}\geq c\langle f,f\rangle_{L^2}, \quad \textrm{for all } f\in L^2([0,T],\R^N)
\ee
due to  Lemma \ref{lemma:lambda positive}. It follows that the operator $\mathbf{D}$ is positive definite with
\be
\langle \mathbf{D}f,f\rangle_{L^2}\geq c\langle f,f\rangle_{L^2}, \quad \textrm{for all } f\in L^2([0,T],\R^N). 
\ee
Therefore, Lemma \ref{lemma:inverse operator} implies that $\mathbf{D}$ is invertible.

It follows that the deterministic $u^*$ in \eqref{eq:deterministic solution} solves \eqref{fred-foc} in Lemma \ref{lemma:fred-foc}. This solution is unique due to Proposition \ref{prop:Fredholm}. Thus, $u^*$ in \eqref{eq:deterministic solution} is the unique maximizer of $J(u)$.
\end{proof}

 \section{Proofs of Theorem \ref{thm:convolution} and Corollary \ref{cor:product of matrix and real function}} \label{sec-pf-mani}
This section is dedicated to the proofs of Theorem \ref{thm:convolution} and Corollary \ref{cor:product of matrix and real function}, which concern convolution kernels. However, one of the main ingredients for these proofs is Proposition \ref{prop:pos semidef kernels}, which gives a new sufficient condition for the nonnegative definite property for a more general class of Volterra kernels. Since this result could be of independent interest for the theory of Volterra equations we assume in the first part of this section that the kernel $G$ is a Volterra kernel.

\subsection{Volterra kernels}
\begin{definition}\label{def:mirrored kernel}
Let $G\in L^2([0,T]^2,\R^{N\times N})$ be a  Volterra kernel. Then the associated mirrored kernel $\smash{\widetilde{G}\in L^2([0,T]^2,\R^{N\times N})}$ is defined as
\be
\widetilde{G}(t,s)\vcentcolon =
\begin{cases}
    G(t,s)\hspace{29.9mm} \text{for }t>s,\\
    \frac{1}{2}(G(t,t)+G(t,t)^\top)\hspace{6.4mm} \text{for }t=s,\\
    G(s,t)^\top\hspace{27.3mm} \text{for }t<s.\\
\end{cases}
\ee
\end{definition}
\begin{remark}
It holds for any Volterra kernel $G\in L^2([0,T]^2,\R^{N\times N})$ that
\be
\widetilde{G}(t,s)=G(t,s)+G(s,t)^\top\quad dt\otimes ds\text{-a.e. on }[0,T]^2,
\ee
and therefore that 
$\langle f, \mathbf{\widetilde{G}}f\rangle_{L^2}=\langle f, (\mathbf{G}+\mathbf{G}^*)f\rangle_{L^2}$ for every $f\in L^2([0,T],\R^N)$. This observation will be crucial for the proof of Proposition \ref{prop:pos semidef kernels}.
\end{remark}
The following proposition provides a sufficient condition \eqref{eq:positive defnite kernel} for the nonnegative definite property of Volterra kernels. Note that \eqref{eq:positive defnite kernel} below is similar to Definition 2 in \cite{alfonsi2016multivariate}, which was used in order to define nonnegative definite convolution kernels.

\begin{prop}\label{prop:pos semidef kernels}
Let $G\in L^2([0,T]^2,\R^{N\times N})$ be a matrix-valued Volterra kernel such that its mirrored kernel $\smash{\widetilde{G}}$ is continuous 
 and satisfies
\be \label{eq:positive defnite kernel}
\sum_{k,l=1}^n x_k^\top\widetilde{G}(t_k,t_l)x_l\geq 0, 
\ee
for any $n\in\N$, $t_1,\dots, t_n\in [0,T]$ and $x_1,\dots,x_n\in\R^N$. Then for any function $f\in L^2([0,T],\R^N)$ the following holds,
\be
\langle f, \mathbf{G}f\rangle_{L^2}
=\frac{1}{2}\langle f, (\mathbf{G}+\mathbf{G}^*)f\rangle_{L^2}
=\frac{1}{2}\langle f, \mathbf{\widetilde{G}}f\rangle_{L^2} \geq 0.
\ee
\end{prop}
\noindent In order to prove Proposition \ref{prop:pos semidef kernels} we will need the following auxiliary lemma. 
\begin{lemma}\label{lemma:equivalence of real and complex}
Let $G :[0,T]\times [0,T]\to \R^{N\times N}$ be a map satisfying  $G(t,s)=G(s,t)^\top$ for all $t,s\in [0,T]$. Then it holds that
\be  \label{eq:positive definite kernel real}
\sum_{k,l=1}^n x_k^\top G(t_k,t_l)x_l\geq 0 \ \text{for any }n\in\N,\  t_1,\dots, t_n\in [0,T],\  x_1,\dots,x_n\in\R^N  \
\ee
 if and only if 
\be \label{eq:positive definite kernel complex}
\sum_{k,l=1}^n \overline{z_k}^\top G(t_k,t_l)z_l\geq 0  \ \text{for any }n\in\N,\  t_1,\dots, t_n\in [0,T],\  z_1,\dots,z_n\in\C^N.   
\ee
\end{lemma}
\begin{proof}
Let $n\in\N$, $t_1,\dots, t_n\in [0,T]$, $z_1,\dots,z_n\in\C^N$, and assume \eqref{eq:positive definite kernel real}. Using \eqref{eq:positive definite kernel real} we get  
\bq
&&\sum_{k,l=1}^n \overline{z_k}^\top G(t_k,t_l)z_l \\
&&=\sum_{k,l=1}^n \big(\operatorname{Re}(z_k)-i\operatorname{Im}(z_k)\big)^\top  G(t_k,t_l)\big(\operatorname{Re}(z_l)+i\operatorname{Im}(z_l)\big)  \\
&&
\geq -i\sum_{k,l=1}^n \operatorname{Im}(z_k)^\top  G(t_k,t_l)\operatorname{Re}(z_l) 
+ i\sum_{k,l=1}^n \operatorname{Re}(z_k)^\top  G(t_k,t_l)\operatorname{Im}(z_l)  \\
&&=-i\sum_{k,l=1}^n \operatorname{Im}(z_k)^\top  G(t_k,t_l)\operatorname{Re}(z_l) 
+ i\sum_{k,l=1}^n \operatorname{Im}(z_l)^\top  G(t_k,t_l)^\top\operatorname{Re}(z_k) \\
&&=0,  
\eq
since $G(t_k,t_l)^\top=G(t_l,t_k)$ by assumption this proves \eqref{eq:positive definite kernel complex}. The reverse direction is trivial.
\end{proof}
\begin{proof}[Proof of Proposition \ref{prop:pos semidef kernels}]
 The main idea of the proof is to apply a generalized version of Mercer's theorem for matrix-valued kernels (see \cite{devito2013extension}). In order to apply this result, first notice that $[0,T]$ is a separable metric space with respect to the Euclidean distance and that the Lebesgue measure defined on $\mathcal{B}([0,T])$ is finite. Next, by assumption, the mirrored kernel $\smash{\widetilde{G}}:[0,T]^2\to\R^{N\times N}$ associated with $G$ is continuous and satisfies condition \eqref{eq:positive defnite kernel}. Therefore by Lemma \ref{lemma:equivalence of real and complex} it also satisfies \eqref{eq:positive definite kernel complex}. Moreover, since $[0,T]^2$ is compact and $\smash{\widetilde{G}}$ is continuous, it holds that 
\be \label{tr-bnd}
\int_0^T \operatorname{Tr} \widetilde{G}(t,t)dt\leq TN \sup_{
1\leq i \leq N}\sup\limits_{0\leq t\leq T}   \widetilde{G}_{ii}(t,t)   <\infty,
\ee
where Tr denotes the trace of a matrix in $\R^{N\times N}$. Note that $G$ must be nonnegative on the diagonal, this can be seen for example from the expansion \eqref{g-t-exp} used later.   
Next, define the linear integral operator
\be
\mathbf{\widetilde{G}}_{\mathbb{C}}:L^2([0,T],\mathbb{C}^N)\to L^2([0,T],\mathbb{C}^N), \quad (\mathbf{\widetilde{G}}_{\mathbb{C}}f)(t)= \int_0^T \widetilde{G}(t,s)f(s)ds,
\ee
which is well-defined and bounded (see Theorem 9.2.4 in \cite{gripenberg1990volterra}), and denote by $\text{ker}(\smash{\mathbf{\widetilde{G}}_{\mathbb{C}}})$ its kernel.
Since \eqref{eq:positive definite kernel complex} and \eqref{tr-bnd} hold, we can apply a generalized version of Mercer's theorem for matrix-valued kernels (see Theorem 4.1 in  \cite{devito2013extension}), 
which, in combination with the auxiliary Theorem A.1 in \cite{devito2013extension}, implies that there exists a countable orthonormal basis $\{\varphi_m\}_{m\in I}$ of $\text{ker}(\smash{\mathbf{\widetilde{G}}_{\mathbb{C}}})^\perp\subset L^2([0,T],\mathbb{C}^N)$ of continuous eigenfunctions of $\smash{\mathbf{\widetilde{G}}_{\mathbb{C}}}$ with a corresponding family $\{\sigma_m\}_{m\in I}\subset(0,\infty)$ of positive eigenvalues such that
\be \label{g-t-exp}
\widetilde{G}_{ij}(t,s)=\sum_{m\in I}\sigma_m \overline{\varphi_m^i(t)}\varphi_m^j(s),\quad \textrm{for all } (s,t)\in [0,T]^2, \ i,j\in\{1,\dots,N\},
\ee
where the series converges uniformly on $[0,T]^2$. Let $f\in L^2([0,T],\R^N)$. From \eqref{g-t-exp} it follows that 
\be \label{t1}
\begin{aligned}
\langle f, \mathbf{\widetilde{G}}f\rangle_{L^2} 
&=\int_0^T\int_0^T \sum_{i,j=1}^N f^i(t) \Big(\sum_{m\in I}\sigma_m \overline{\varphi_m^i(t)}\varphi_m^j(s)\Big) f^j(s)dsdt \\ 
&=\int_0^T\int_0^T \sum_{m\in I}\sigma_m \sum_{i,j=1}^N f^i(t) \overline{\varphi_m^i(t)}\varphi_m^j(s) f^j(s)dsdt \\
&=\int_0^T\int_0^T \int_{I} \sigma_m \sum_{i,j=1}^N f^i(t) \overline{\varphi_m^i(t)}\varphi_m^j(s) f^j(s)d\mu (m)dsdt,
\end{aligned} 
\ee
where $\mu$ is the counting measure on the countable index set of the eigenfunctions $I$. In particular, $\mu$ is $\sigma$-finite. 

Next, we would like to switch the order of integration on the right-hand side of \eqref{t1} using Fubini's theorem. In order to do that, we show that the following integral is finite, by using Young's inequality, \eqref{tr-bnd} and \eqref{g-t-exp}, 
\bn
&& \int_0^T\int_0^T \int_{I} \big|\sigma_m \sum_{i,j=1}^N f^i(t) \overline{\varphi_m^i(t)}\varphi_m^j(s) f^j(s)\big| d\mu (m)dsdt \\
&&
\leq \int_0^T\int_0^T \sum_{i,j=1}^N \big|f^i(t) \big|   \int_{I} \sigma_m \big(\big|\varphi_m^i(t)\big|^2+\big|\varphi_m^j(s) \big|^2\big) d\mu (m)\big|f^j(s)\big|dsdt \\
&&= \int_0^T\int_0^T \sum_{i,j=1}^N \big|f^i(t) \big| \big(\widetilde{G}_{ii}(t,t)+\widetilde{G}_{jj}(s,s)\big)  \big|f^j(s)\big|dsdt \\ 
&&
\leq  2\Big(\sup_{1\leq i \leq N} \sup_{0\leq t\leq T}\wt G_{ii}(t,t)\Big) \int_0^T\int_0^T \sum_{i,j=1}^N \big|f^i(t) \big| \big|f^j(s)\big|dsdt \\ 
&&<\infty.
\en
Let $\mathbf{1}_N$ denote the $N\times N$ matrix of ones in all entries, which is nonnegative definite as its eigenvalues are $N$ with multiplicity $1$ and $0$ with multiplicity $N-1$. 
Moreover, for every $m\in I$ we define 
\be
g_m:[0,T]\to\mathbb{C}^N,\quad g_m(t)\vcentcolon
=\begin{pmatrix} \varphi_m^1(t) f^1(t) \\ \vdots \\ \varphi_m^N(t) f^N(t) \end{pmatrix} \;,
\ee
which is the Hadamard product of $\varphi_m$ and $f$.
Then \eqref{t1} and Fubini's theorem imply 
\be
\begin{aligned}
\langle f, \mathbf{\widetilde{G}}f\rangle_{L^2}
&=\int_{I}^{\ } \int_0^T\int_0^T  \sigma_m \sum_{i,j=1}^N f^i(t) \overline{\varphi_m^i(t)}\varphi_m^j(s) f^j(s)dsdtd\mu (m) \\
&=\int_{I}^{\ }\int_0^T\int_0^T \sigma_m\ \overline{g_m(t)}^\top \mathbf{1}_N \ g_m(s) dsdtd\mu (m) \\ 
&=\sum_{m\in I} \sigma_m \Big(\overline{\int_0^T  g_m(t)dt}\Big)^\top  \mathbf{1}_N\  \Big(\int_0^T g_m(s) ds\Big)\geq 0.
\end{aligned}
\ee
This finishes the proof.
\end{proof}
\begin{remark}
In the proof of Proposition \ref{prop:pos semidef kernels}, we use Fubini's theorem and the counting measure $\mu$ to justify interchanging integration and summation, since the uniform convergence of 
\be
\widetilde{G}_{ij}(t,s)=\sum_{m\in I}\sigma_m \overline{\varphi_m^i(t)}\varphi_m^j(s),
\ee
on $[0,T]^2$  does not imply the uniform convergence of
\be
\sum_{m\in I}\sigma_m f^i(t) \overline{\varphi_m^i(t)}\varphi_m^j(s) f^j(s)
\ee
on $[0,T]^2$ for any $f\in L^2([0,T],\R^N)$ as $f$ may be unbounded.
\end{remark}

\subsection{Convolution kernels}
For the remainder of this section we restrict our discussion to convolution kernels of the form $G(t,s)\vcentcolon =\mathds{1}_{\{t\geq s\}}H(t-s)$ as in Theorem \ref{thm:convolution}. 
\begin{proof}[Proof of Theorem \ref{thm:convolution}] Without loss of generality, we can assume that 
\be
H(T)=\lim_{t\uparrow T}H(t)\in\R^{N\times N},
\ee
where the limit always exists, since $H$ is continuous, nonincreasing, nonnegative and symmetric on $(0,T)$ by assumption.

\textbf{Case 1.} We assume that $H$ is bounded. Then we can also assume without loss of generality that 
\be
H(0)=\lim_{t \dr 0}H(t)\in\R^{N\times N},
\ee
where the limit exists by a similar argument.
Now we continuously extend $H$ to $[0,\infty)$ by defining
\be \label{ext}
\overline{H}(t)\vcentcolon 
=\begin{cases}
    H(t)\hspace{10mm}\text{for }t<T,\\
    H(T)\hspace{8.5mm}\text{for }t\geq T.
\end{cases}
\ee
Then $\overline{H}$ is continuous, symmetric, nonnegative, nonincreasing and convex on $[0,\infty)$. Thus, it follows from Theorem 2 in \cite{alfonsi2016multivariate} that the mirrored kernel associated with the Volterra kernel $\overline{G}(t,s)\vcentcolon =\mathds{1}_{\{t\geq s\}}\overline{H}(t-s)$ on $[0,\infty)^2$ satisfies \eqref{eq:positive definite kernel complex}. Therefore, the continuous mirrored kernel $\smash{\wt{G}}$ associated with $G$ satisfies \eqref{eq:positive definite kernel complex} as well. This allows us to apply Proposition \ref{prop:pos semidef kernels} to conclude that for every $f\in L^2([0,T],\R^N)$,
\be
\int_0^T\int_0^t f(t)^\top H(t-s)f(s)dsdt=\langle f,\mathbf{G}f\rangle_{L^2}
=\frac{1}{2}\langle f, \mathbf{\widetilde{G}}f\rangle_{L^2}\geq 0.
\ee

\textbf{Case 2.} We consider the case where $H$ is unbounded. Then it follows from the assumptions of the theorem that $H$ must have a singularity at $t=0$, i.e. the limit of $H(t)$ as $t$ decreases to $0$ does not exist in $\R^{N\times N}$. Let $\overline{H}$ denote the extension of $H$ to $[0,\infty)$ as in \eqref{ext}. Then it follows that $\overline{H}$ is  nonincreasing, continuous, convex, nonnegative and symmetric on $(0,\infty)$. Define the sequence of convolution kernels $(H_m)_{m\geq 1}$  by
\be \label{h-m}
H_m:[0,T]\to\R^{N\times N},\quad H_m(t)\vcentcolon=\overline{H}(t+m^{-1}),\quad\textrm{for }m\geq 1,
\ee
and notice that it converges pointwise to $H$ on $(0,T]$, as $\overline{H}$ is continuous on $(0,\infty)$. Furthermore, for every $m\geq 1$, $H_m$ is nonincreasing, convex, nonnegative, symmetric, continuous and thus bounded on $[0,T]$, so that the proof of Case 1 applies to them. 

In order to proceed we notice that if $C=(C_{ij})_{i,j\in\{1,\dots,N\}}$ is symmetric nonnegative definite matrix, then all of its diagonal entries $C_{ii}$ are nonnegative and we have 
\be
|C_{ij}|\leq \frac{1}{2}\big(C_{ii}+C_{jj}\big)\leq \max\limits_{1\leq i \leq N}C_{ii},
\ee
where we have used the fact that $x^\top Cx\geq 0$ for $x=e_i\pm e_j$ and the symmetry of $C$. This implies the second inequality of the following, 
\be \label{bndC}
\big(\max\limits_{1\leq i\leq N}C_{ii}\big)^2\leq\|C\|^2\leq N^2\big(\max\limits_{1\leq i\leq N}C_{ii}\big)^2,
\ee
where the first inequality follows from the definition of the Frobenius norm. 

Let $f\in L^2([0,T],\R^N)$. Using, in order of appearance below, the submultiplicativity of the Frobenius norm, Definition \eqref{h-m} and Young's inequality, the right-hand side of \eqref{bndC}, the facts that the diagonal set $\Delta:=\{(t,s)\in[0,T]^2:t=s\}$ has $\mathbb{R}^2$-Lebesgue measure zero and that the diagonal entries of $\overline{H}(t)$ are nonincreasing on $(0,\infty)$, it holds for all $m\geq 1$, 
\be
\begin{aligned} \label{fe1}
&\int_0^T\int_0^T  \big| f(t)^\top \mathds{1}_{\{t\geq s\}}H_m(t-s)f(s)\big|dsdt \\ 
&\leq \int_0^T\int_0^T \mathds{1}_{\{t\geq s\}} \big\|f(t)\big\| \big\|H_m(t-s)\big\| \big\|f(s)\big\|dsdt \\ 
&\leq \int_0^T\int_0^T\mathds{1}_{\{t\geq s\}} \Big( \big\|\overline{H}(t-s+m^{-1})\big\|^2+ \big\|f(t)\big\|^2 \big\|f(s)\big\|^2 \Big)dsdt \\ 
&\leq N^2\int_0^T\int_0^T\mathds{1}_{\{t\geq s\}} \Big( \big( \max\limits_{1\leq i \leq N} \overline{H}_{ii}(t-s+m^{-1})\big)^2+ \big\|f(t)\big\|^2 \big\|f(s)\big\|^2 \Big)dsdt \\ 
&\leq N^2\int_0^T\int_0^T\mathds{1}_{\{t> s\}} \Big(\big( \max\limits_{1\leq i \leq N} \overline{H}_{ii}(t-s)\big)^2+ \big\|f(t)\big\|^2 \big\|f(s)\big\|^2\Big)dsdt \\ 
&\leq N^2\int_0^T\int_0^T\mathds{1}_{\{t> s\}}\Big( \big\| \overline{H}(t-s)\big\|^2+ \big\|f(t)\big\|^2 \big\|f(s)\big\|^2 \Big)dsdt \\
&\leq N^2\int_0^T\int_0^T  \big\|G(t,s)\big\|^2dsdt +N^2 \big\|f\big\|_{L^2}^4 \\
&<\infty,
\end{aligned}
\ee
where we have used the left-hand side of \eqref{bndC}, the fact that $G(t,s)=\mathds{1}_{\{t\geq s\}}H(t-s)=\mathds{1}_{\{t\geq s\}}\overline{H}(t-s)$ for $(s,t)\in[0,T]^2$, and finally that $G$ is in $L^2([0,T]^2,\R^{N\times N})$ in the last three inequalities. 

Now the sequence of functions 
\be
\big(f(t)^\top \mathds{1}_{\{t\geq s\}}H_m(t-s)f(s)\big)_{m\geq 1}
\ee
defined on $[0,T]^2$ converges pointwise on $[0,T]^2\setminus \Delta$ and thus almost everywhere to
\be
f(t)^\top \mathds{1}_{\{t\geq s\}}H(t-s)f(s).
\ee
Therefore, it follows from \eqref{fe1} that dominated convergence can be applied to get,  
\be
\int_0^T\int_0^t f(t)^\top H(t-s)f(s)dsdt=\lim\limits_{m\to\infty}\int_0^T\int_0^t f(t)^\top H_m(t-s)f(s)dsdt\geq 0,
\ee
which completes the proof.
\end{proof}

\begin{proof} [Proof of Corollary \ref{cor:product of matrix and real function}]
First, note that the function given by
\be
[0,T]\to\R,\quad t\mapsto\int_0^t\phi(t-s)^2ds
\ee
is the convolution of the integrable function $\phi^2$ and the constant function $\mathds{1}_{[0,T]}$. It follows that the convolution is integrable again (see Theorem 2.2.2 (i) in \cite{gripenberg1990volterra}), i.e.
\be
\int_0^T\int_0^t\phi(t-s)^2dsdt<\infty.
\ee
Therefore the Volterra kernel $\overline{G}(t,s) =\mathds{1}_{\{t\geq s\}}H(t-s)$ is in $L^2([0,T]^2,\R^{N\times N})$ for $H$ as in \eqref{h-cfm}. Next, for any $x\in\R^N$ we have that $x^\top Cx\geq 0$. This implies that the function \be
t\mapsto x^\top H(t) x = x^\top Cx \phi(t) 
\ee
is nonincreasing and convex on $(0,T)$ because $\phi$ is nonincreasing and convex on $(0,T)$ by assumption. Finally, the matrix 
\be
H(t)=C \phi(t) 
\ee
is symmetric nonnegative definite for every $t\in (0,T)$, since $\phi$ is nonnegative and $C$ is symmetric nonnegative definite. Thus, the conditions of Theorem \ref{thm:convolution} are satisfied.
\end{proof}

\section{Proof of Lemma \ref{lemma:examples}}\label{sec-proof-lem-ker}
\begin{proof}[Proof of Lemma \ref{lemma:examples}]
Note that all propagator matrices from Example \ref{ex:propagators} except for the singular power-law kernel in Example \ref{ex:propagators}(ii) and the non-convolution kernel in Example \ref{ex:propagators}(v) are of the general form 
\be
G(t,s)=Q^T \mathds{1}_{\{t\geq s\}}\operatorname{diag}\big(g_1(t-s),\dots,g_N(t-s)\big)Q
\ee
for an invertible matrix $Q\in\R^{N\times N}$ and nonnegative, nonincreasing, convex functions $g_1,\dots,g_N:[0,T]\to\R$. For the matrix exponential kernel, this follows from Example 2 in \cite{alfonsi2016multivariate}. For the other kernels, an eigendecomposition of the corresponding symmetric nonnegative definite matrix  $C$ yields the desired representation. Next, in order to see that $G\in\mathcal{G}$, we proceed by verifying the conditions of Theorem \ref{thm:convolution} for the kernel $H(t)\vcentcolon = Q^\top \operatorname{diag}\big(g_1(t),\dots,g_N(t)\big)Q$. Let $x\in\R^N$ and $y\vcentcolon =Qx$. Then it follows that
\be
x^\top H(t) x= y^\top \operatorname{diag}\big(g_1(t),\dots,g_N(t)\big)y=\sum_{i=1}^N y_i^2 g_i(t)
\ee
is nonnegative, nonincreasing and convex on $[0,T]$. Moreover $H(t)$ is symmetric for any $t\in [0,T]$. Therefore, Theorem \ref{thm:convolution} applies, since $G$ is bounded and thus in $L^2([0,T]^2,\R^{N\times N})$.

The singular power-law kernel in Example \ref{ex:propagators}(ii) given by
\be 
G(t,s)=\mathds{1}_{\{t>s\}}(t-s)^{-\beta}C
\ee
for $\beta\in(0,\frac{1}{2})$ and nonnegative definite $C\in\R^{N\times N}$ is in $\mathcal{G}$ due to Corollary \ref{cor:product of matrix and real function}.

For the non-convolution kernel in Example \ref{ex:propagators}(v) we follow similar lines as in the proof of Lemma \ref{lemma:F+F^* nonnegative def}, using the fact that $H:[0,T]\to\R$ is nonnegative definite (see Example 2.7 in \cite{GSS}) to get
\bn
&&\int_{0}^T\int_0^T f(t) \alpha (T-t) \mathds 1_{\{t>s\}}H(t-s)  f(s)  dsdt  \\ 
 &&=  \alpha \int_0^T \int_{0}^{T-r}\int_0^{T-r}  f(t) \mathds 1_{\{t>s\}}H(t-s)  f(s)   dsdt dr\\
&& \geq 0, \quad \textrm{for all }f\in L^2([0,T],\R). 
\en
Hence $G$ is a product of a nonnegative, nonincreasing, convex function taking values in $\R$ and a nonnegative definite matrix $C\in\R^{N\times N}$, so the result follows from the argument used for examples (i)--(iv) above. 
\end{proof}

\textbf{Acknowledgment}. We are very grateful to the anonymous referees and the editors for their careful reading of the manuscript and for their insightful comments and suggestions, which significantly improved the paper.

\textbf{Data Availability Statement.}
Data sharing is not applicable to this article as no new data were created or analyzed in this study.

\end{document}